\documentclass[11pt, letterpaper]{article}

\usepackage{vasilis}
\newcommand{\opt}{{\rm OPT}}
\newcommand{\alg}{{ALG}}
\newcommand{\optalg}{OPTALG}

\DeclareUnicodeCharacter{03BB}{$\lambda$}

\makeatletter
\newcommand{\newreptheorem}[2]
{\newtheorem*{rep@#1}{\rep@title}
\newenvironment{rep#1}[1]
{\def\rep@title{#2 \ref*{##1}}
\begin{rep@#1}}{\end{rep@#1}}}
\makeatother

\newreptheorem{theorem}{Theorem}
\newreptheorem{lemma}{Lemma}
\newreptheorem{claim}{Claim}
\newreptheorem{observation}{Observation}
\newreptheorem{proposition}{Proposition}

\allowdisplaybreaks

\title{Prophet Inequalities for Cost Minimization\thanks{Dept.\ of Computer Science, University of Illinois at Urbana-Champaign, IL 61801. {\tt \{livanos3, rutameht\}@illinois.edu}.}
}

\author{
Vasilis Livanos %
\and
Ruta Mehta
}

\date{\today}

\begin{document}
\maketitle
\thispagestyle{empty}

\begin{abstract} 
Prophet inequalities for \emph{rewards maximization} are fundamental to optimal stopping theory with extensive applications to mechanism design and online optimization. We study the \emph{cost minimization} counterpart of the classical prophet inequality: a decision maker is facing a sequence of \emph{costs} $X_1, X_2, \dots, X_n$ drawn from known distributions in an online manner and \emph{must} ``stop'' at some point and take the last cost seen. The goal is to compete with a ``prophet'' who can see the realizations of all $X_i$'s upfront and always select the minimum, obtaining a cost of $\E\brk{\min_i X_i}$.

We observe that if the $X_i$'s are not identically distributed, no strategy can achieve a bounded approximation, even for random arrival order and $n = 2$. This leads us to consider the case where the $X_i$'s are \emph{independent} and \emph{identically distributed (I.I.D.)}. For the I.I.D. case, we show that if the distribution satisfies a mild condition, the optimal stopping strategy achieves a (distribution-dependent) constant-factor approximation to the prophet's cost. Moreover, for MHR distributions, this constant is at most $2$. All our results are tight. We also demonstrate an example distribution that does not satisfy the condition and for which the competitive ratio of any algorithm is infinite. 

Turning our attention to single-threshold strategies, we design a threshold that achieves a $\bigO{\polylog{n}}$-factor approximation, where the exponent in the logarithmic factor is a distribution-dependent constant, and we show a matching lower bound. Finally, we note that our results can be used to design approximately optimal posted price-style mechanisms for \emph{procurement auctions} which may be of independent interest.

Our techniques utilize the \emph{hazard rate} of the distribution in a novel way, allowing for a fine-grained analysis which could find further applications in prophet inequalities.
\end{abstract}

\medskip
\medskip
\medskip
\noindent

\clearpage
\setcounter{page}{1}

\newcommand{\rh}[1]{\textcolor{blue}{#1}}
\section{Introduction}\label{sec:introduction}

The classical prophet inequality due to Krengel, Sucheston, and Garling \cite{kren-such} concerns the setting where one is presented with \emph{take-it-or-leave-it} rewards $X_1, \dots, X_n$ in an online manner, drawn independently from known distributions, and can ``stop’’ at any point and collect the last reward seen. Given that the distributions are known, the inequality ensures the existence of a {\em stopping strategy $S$} (online algorithm) for any arrival order of the random variables, with expected reward at least half that of a {\em prophet} who can see the realizations of all the $X_i$'s upfront (offline optimum), {\em i.e.,}  $\E\brk{S}\ge \frac{1}{2} \E\brk{\max_i X_i}$. This result, and its variations and generalizations, have found extensive applications to online optimization and mechanism design, particularly, in the design of simple yet approximately optimal sequential posted price mechanisms, both online and offline, for {\em revenue (rewards) maximization} while {\em selling items} \cite{haji, ChawlaHMS, chawla07, klein-wein} (see Section~\ref{sec:related-work} for a detailed discussion).

However, what if the $X_i$'s are {\em costs} and the goal is {\em cost minimization}, like in the case of {\em procuring items} while {\em minimizing the payment}? For example, consider a house buyer trying to decide when to buy a house in a sellers' market, where houses are selling fast. When a house arrives with its price (cost) listed, she may have to decide the same day whether to buy it or not. Given that the buyer may have only distributional knowledge of future house prices, the goal is to devise a buying strategy so that the price paid is minimized.

Towards this, we study the {\em cost} counterpart of the prophet inequality, where the $X_i$'s represent costs arriving in an online manner, and one {\em must} ``stop'' at some point and select the last cost seen. Note that here the constraint is {\em upwards-closed}, {\em i.e.,} one of the $X_i$'s {\em has to be} selected. In particular, if one makes it to $X_n$, they are forced to pick its realization regardless of how high it is. The goal is to design a stopping strategy (online algorithm) $\alg$ that {\em minimizes the expected cost}, and is comparable to the cost of an all-knowing prophet who can always select the minimum realization and thus has expected cost $\E\brk{min_i X_i}$. For an $\alpha \geq 1$, we say that algorithm $\alg$ achieves an $\alpha$-factor {\em cost prophet inequality}, or is $\alpha$-{\em competitive/approximate}, if
\begin{equation}\label{eq:a-approx}
\E\brk{\alg} \le \alpha \cdot \E\brk{\min_i X_i}.
\end{equation}

For the rewards maximization setting, the competitive ratio of $\nicefrac{1}{2}$ in the classical prophet inequality is achievable through simple single-threshold algorithms~\cite{sam-cahn, klein-wein} of the form ``accept the first $X_i \geq \tau$ for some threshold $\tau$'', and is known to be tight. Furthermore, there exist simple online algorithms that achieve constant-factor approximations even for general multi-dimensional settings with complicated constraints (matroids, matchings, etc) \cite{klein-wein, Alaei14, willma-adversarial, gravin-bipartite, ezra}. Motivated by these works, we ask:

\begin{center}
{\em For the cost minimization setting, can we obtain similar results to the rewards setting?

In particular, what is the factor achieved by the optimal online algorithm, i.e., the best (smallest) possible $\alpha$ by any online algorithm, in the cost prophet inequality of \eqref{eq:a-approx}?
Is the factor achieved a constant? 

Is it achievable by simple single-threshold algorithms?}
\end{center}

In this paper, we study the above questions. At first glance, one may wonder why the cost setting is not equivalent to the classical prophet inequality for reward maximization with {\em negative} $X_i$'s. The main reason is that an algorithm for the rewards setting works with {\em downwards-closed} constraints (see Section~\ref{sec:related-work}) and therefore, if all $X_i$'s are negative, the optimal solution is trivial: the algorithm will not select any $X_i$ and obtain a value of $0$. However, this violates the {\em upwards-closed} constraint of the cost prophet inequality. In fact, this difference turns out to be a crucial one, as we demonstrate that upwards-closed constraints lead to qualitatively different guarantees.

In particular, towards answering the above questions, we obtain tight bounds for the {\em optimal online algorithm} as well as for {\em single threshold algorithms}. Surprisingly, these bounds turn out to be qualitatively different from what is known for the {\em rewards maximization setting}. In what follows, we give an overview of our results and the techniques used.

\subsection{Overview of our Contributions and Techniques}

We first observe that when the $X_i$'s are not identically distributed, no algorithm can achieve {\em any bounded approximation} if the arrival order of the $X_i$'s is adversarial or even random. This was first observed by Esfandiari, Hajiaghayi, Liaghat and Monemizadeh \cite{esf-prophsec}, and we present a simplified version of their counterexample. In particular, $\alpha$ in \eqref{eq:a-approx} can be unbounded, even in the special case of $n=2$ and the distributions with support at most two (see Proposition~\ref{prp:negative-arrival}). This strong negative result leads us to consider the I.I.D. case where the $X_i$'s are drawn independently from the same known distribution.

\paragraph{Independent and Identically Distributed (I.I.D.) Costs.}
In the I.I.D. case, all the $X_i$'s are drawn from a common non-negative known distribution $\cD$. Since single-threshold algorithms have been very successful in providing constant-factor approximations in the rewards setting, we start by asking whether we can achieve similar results for cost prophet inequalities. The intuition behind this is that if $n$ is very large, one could set a single threshold close to $\E[\min_i X_i]$ and with good probability there will be at least one realization below the threshold.

Unfortunately, this intuition turns out to be wrong, even for simple distributions, like the exponential. We present an example in Appendix~\ref{app:counterexamples} explaining why this intuition fails. In fact, in Section~\ref{sec:single-threshold} we show that no single threshold can achieve a constant competitive ratio; we discuss what one can achieve by a single threshold later.

Thus, in our search for a constant-factor competitive algorithm, we study {\em optimal online algorithms}, i.e., algorithms that achieve the smallest possible $\alpha$ in \eqref{eq:a-approx}. It follows that, much like the classical prophet inequality \cite{correa, renato-iid}, it suffices to only consider threshold-based algorithms with oblivious thresholds: set thresholds $\tau_1, \dots, \tau_n$ upfront, and accept the first $X_i \leq \tau_i$. Intuitively, this is because the process is {\em memory-less}, i.e. the decision in the $i$-th round is independent of the past realizations and depends only on the realization of $X_i$ and the distribution of the future costs.

Notice that $\tau_n = + \infty$ because, if for all $i < n$ we have $X_i > \tau_i$, then the algorithm has to accept $X_n$ no matter how high its value is. In that case, it incurs cost of $E_{X \sim \cD} \brk{X}$. Using this fact, it suffices to set $\tau_{n-1} = \E_{X \sim \cD} \brk{X}$, which is the expected cost of the algorithm if it decides not to select $X_{n-1}$. Then, by standard backwards induction, it can be shown that the optimal $\tau_i$ is equal to the expected cost that the optimal algorithm incurs when there are $n - i$ remaining random variables to be drawn from $\cD$ (see Section \ref{sec:opt-thrs}). This fact verifies the natural intuition that we should stop and select $X_i$ only if its realization is better (smaller) than the expected future cost; we include the formal proof for the sake of completeness.

Note that the optimal online algorithm may incur the expected value as a cost with positive probability, since it always has to select a cost. If the expected value is not finite and the prophet's cost is finite, then no non-trivial factor is possible. This fact prevents any bounded factor approximation for {\em all} distributions. We formalize this in Section~\ref{sec:preliminaries} through a simple example with $n = 2$ (due to Lucier \cite{brendan-personal}).

A natural next step is to search for the largest class of \emph{tractable distributions}. We approach this goal via a new technique that utilizes the hazard rate of the distribution. We identify a somewhat surprising property of the hazard rate that allows us to characterize a hierarchy of tractable distributions. After introducing the hazard rate and its properties, we discuss how our analysis crucially use it to obtain the tight approximation factors. To the best of our knowledge, this is the first use of the hazard rate as an analysis tool in prophet inequality-type problems, and we believe our approach may be of independent interest in analyzing both cost and rewards prophet inequalities.

\paragraph{Hazard Rate.}
For a given distribution $\cD$ with probability density and cumulative distribution functions $f$ and $F$ respectively, the {\em hazard rate} of $\cD$ is defined for all $x$ in the support of $\cD$ as $h(x) \triangleq \frac{f(x)}{1 - F(x)}$. Also referred to as the {\em failure rate}, it is a fundamental quantity within several fields of economics and mathematics and has found a lot of applications in survival analysis \cite{hazard-survival}, reliability theory \cite{hazard-reliability}, pricing \cite{hartline-simple, giannakopoulos-mhr,babaioff-pricing,bhattacharya-pricing,cai-pricing,daskalakis-mhr,revenue-sample-mhr,giannakopoulos-markets} and even forensic analysis \cite{hazard-forensic}. For our results, we utilize the integral of the hazard rate, $H(x) = \int_0^x h(z) \dif z$, which is called the {\em cumulative hazard rate} of $\cD$.
\medskip
\medskip

\noindent{\bf Analysing the Optimal (Online) Algorithm.}
We express both the prophet's cost and the expected cost of the optimal algorithm, thus also the competitive ratio, as integrals that depend only on $H(x)$ (Observation~\ref{obs:betan-hazard} and Lemma \ref{lem:G-recurrence}). To proceed, we need an explicit expression for $H$. One solution is to consider a series expansion of $H$ around $\alpha$, the infimum of the support of $\cD$. The choice of $\alpha$ is important since $H(\alpha)=0$ by definition, and hence the Taylor series will not have a constant term. Although, we assume $\alpha = 0$, our results easily extend to distributions for which $\alpha \neq 0$ (see Remark~\ref{rmk:support-point}).
\medskip

{\em Entire Distributions.}
It turns out that the Taylor series representation is pretty restrictive since it fails to capture distributions beyond those with monotone hazard rate (MHR). We instead turn to a generalization of Taylor series, the {\em Puiseux series} of a function (for more information see \cite{puiseux-book}). The only difference for us is that Puiseux series allow for fractional exponents in the indeterminate, and this allows us to cover a substantially larger class of distributions. We call these distributions {\em Entire}\footnote{The term is analogous to the notion of entire functions, which are functions that are analytic everywhere, i.e. their Taylor series converges everywhere.}, and our results hold for all such distributions. This class includes almost all commonly used distributions, including the uniform, exponential, Gaussian, Weibull, Rayleigh, arcsine, beta and gamma distributions, among many others. 

A distribution $\cD$ is called \emph{Entire} if the cumulative hazard rate function $H$ of $\cD$ has a Puiseux series expansion $H(x) = \sum_{i = 1}^\infty {a_i x^{d_i}} \neq 0$, and this series is convergent everywhere in the support of $\cD$. This assumption is not just a technicality, since the equal-revenue distribution mentioned earlier has $H(x) = \log{x}$, and its Puiseux series around $x = 1$ only converges in $[1, 2]$ instead of $[1, +\infty)$. Thus the equal-revenue distribution is not Entire and, as we show, no algorithm can obtain a finite competitive ratio for it (see Theorem~\ref{thm:negative-arrival-equal-revenue}).
\medskip

Using the Puiseux series of $H$, we are able to analyze the two integrals. Specifically, through an intricate analysis, we are able to compute the prophet's cost $\E[\min_i X_i]$ exactly (Lemma~\ref{lem:poly-hazard-beta}).  Somewhat surprisingly, the value of $\E[\min_i X_i]$ is dominated only by lowest degree $d_1$ and corresponding coefficient $a_1$ in the series of $H$. We note that this result may be of independent interest. The algorithm's cost can then be computed in a similar manner, using a recursive analysis. Putting everything together, the optimal algorithm is $\lambda(d_1)$-competitive, where $\lambda(d_1)$ is a decreasing function of $d_1$, and the competitive ratio cannot be improved. Interestingly, the approximation depends (inversely) only on $d_1$. Intuitively, this is because $H$ grows rapidly as $d_1$ increases, and thus $\cD$ has a less heavy tail which leads to a better approximation. It turns out that $\lambda(d_1)$ depends on the {\em Gamma function}, an extension of the factorial function over the reals. For its definition, see Section~\ref{sec:preliminaries}.

\begin{theorem}\label{thm:hazard-constant}
For the I.I.D. setting under any given non-negative Entire distribution $\cD$, for large enough $n$, there exists a $\lambda(d)$-factor cost prophet inequality, where
\[
\lambda(d) = \frac{\prn{1+1/d}^{1/d}}{\Gamma\prn{1+1/d}},
\]
$d$ is the smallest degree of the Puiseux series of $H$, and $\Gamma(\cdot)$ is the Gamma function.

Moreover, this constant is tight for the distribution with cumulative hazard rate $H(x) = x^d$.
\end{theorem}

To understand how $\lambda(d)$ grows with $d$, consider Stirling's approximation for the Gamma function,
$
\Gamma(z) \approx \frac{\sqrt{2 \pi}}{z} \prn{\frac{z}{e}}^z.
$
Replacing this in the expression of $\lambda(d)$, we have
\[
\lambda(d)=\frac{\prn{1+1/d}^{1/d}}{\Gamma\prn{1+1/d}} \approx \frac{\prn{1+1/d}^{1/d}}{\frac{\sqrt{2 \pi}}{1 + 1/d} \prn{\frac{1+1/d}{e}}^{1+1/d}} = \frac{e}{\sqrt{2 \pi}} \: \: e^{1/d}.
\]
Thus the dependence of $\lambda(d)$ on $d$ is (approximately) inversely exponential.

\paragraph{MHR Distributions.} Distributions with monotonically increasing hazard rate have been extensively studied in the mechanism design literature due to their sought after properties and applications (e.g., see \cite{giannakopoulos-mhr, babaioff-pricing, bhattacharya-pricing, cai-pricing, daskalakis-mhr, revenue-sample-mhr, giannakopoulos-markets, hartline-simple}). These are known as {\em monotone hazard rate (MHR)} distributions. For such distributions, we are able to show that $d \geq 1$, and since $\lambda(1) = 2$, we show that the optimal algorithm is $2$-factor competitive. In addition, we show that the factor of $2$ is tight for the exponential distribution, which has constant hazard rate.

\begin{theorem}\label{thm:mhr-constant}
For every Entire MHR distribution, there exists a $2$-competitive cost prophet inequality, for large enough $n$.

This factor is tight, since there is no $\prn{2 - \eps}$-cost prophet inequality for any $\eps > 0$ for the exponential distribution, which has constant hazard rate.
\end{theorem}

\paragraph{Single-Threshold Algorithms.} Given the success of single-threshold algorithms in the classical rewards setting, where they are able to achieve the best possible competitive ratio, we ask if there exists a single-threshold algorithm that achieves a constant-factor competitive ratio for the cost prophet inequality setting as well. The answer turns out to be negative. We show that, for Entire distributions, no single-threshold algorithm can achieve a better than poly-logarithmic competitive ratio. We also obtain a matching upper bound. In particular, given an Entire distribution $\cD$, by analyzing how the competitive ratio of a single-threshold algorithm grows as $n$ increases, we design a threshold $T$ such that the algorithm that selects the first $X_i \leq T$ for $i < n$ and $X_n$ otherwise, yields a $\bigO{\polylog{n}}$-factor cost prophet inequality. 
Here, the power in the poly-logarithmic factor inversely depends on the smallest degree of the Puiseux series of $H$.

\begin{theorem}\label{thm:single-threshold}
Given $X_1, \dots, X_n$ drawn independently from a non-negative Entire distribution $\cD$, there exists a single-threshold algorithm that is  $\bigO{\polylog{n}}$-competitive, for large enough $n$. Moreover, this factor is tight, i.e. there exist distributions for which no single-threshold algorithm is $\ltlo{\polylog{n}}$-competitive.
\end{theorem}

\begin{remark}
The assumption in our theorems that $n$ should be large enough is a technicality and we believe our theorems hold for all $n \geq 1$. We show that this intuition indeed holds for distributions with $H(x) = x^d$ where $d > 0$, i.e. the competitive ratio increases with $n$ (Lemma~\ref{lem:ratio-incr}). Empirical analysis of several other distributions also verifies this intuition.
\end{remark}

\paragraph{Application to Mechanism Design.} Finally we note that, similar to the extensive application of classical prophet inequalities in designing simple yet approximately optimal posted-price mechanisms for selling items (see Section~\ref{sec:related-work}), our algorithms and results for the {\em cost prophet inequality} can be used for the design and analysis of posted-price-style mechanisms for {\em procuring items}.

Consider a procurement auction (also known as a reverse auction), in which the auctioneer (buyer) wants to procure a single item sold by $n$ different sellers, with an I.I.D. distribution governing the sellers' valuation for selling the item to the auctioneer or, in other words, the sellers' costs/price. If the sellers arrive in an online manner with take-it-or-leave-it offers, for example as is the case in a seller's housing market, then the standard reduction of a posted-price mechanism to a prophet inequality due to Hajiaghayi, Kleinberg and Sandholm \cite{haji} applies directly to the cost setting, if one wants to minimize the social cost.

To minimize the {\em procurement price } paid by the buyer (auctioneer), 
one simply needs to use the 
the {\em virtual costs} $\phi(c) = c + \frac{F(c)}{f(c)}$, since Myerson's optimal auction \cite{myerson-optimal-auction} applies to any single-parameter environment. This holds only if $\cD$ is a regular distribution (a class of distributions which includes MHR distributions among others). For non-regular $\cD$, one simply needs to ``iron'' the social cost function, just as in the classical rewards setting, and proceed similarly afterwards. For more details on this see \cite{myerson-virtual}~(Theorem 1).
\medskip

\subsection{Related Work}\label{sec:related-work}

Given the intractability of the optimal (revenue-maximizing) mechanisms for selling items \cite{pricing-lotteries, hart-nisan-menu, daskalakis-intractability-1, daskalakis-intractability-2}, the focus turned to designing {\em approximately optimal} yet simple mechanisms where prophet inequalities for {\em reward maximization} have been extensively studied. The works of Hajiaghayi, Kleinberg and Sandholm \cite{haji} and Chawla, Hartline, Malec and Sivan \cite{ChawlaHMS} pioneered the use of prophet inequalities to analyze (sequential) posted price mechanisms for {\em selling items}. Specifically, \cite{haji} observed that the problem of designing posted price mechanisms that maximize welfare can be reduced to an appropriate optimal stopping theory problem, and this was extended to revenue-maximizing posted price mechanisms in \cite{ChawlaHMS}. This result led to a significant effort to understand how the expected revenue of an optimal posted price mechanism compares to that of the optimal auction \cite{chawla07, qiqi, BlumHolen08, Adamczyk17, Alaei14, feldman-combinatorial, Dutting2, kesselheim-mhr-ppm, Dobzinski-Comb-Auc, Dobzinski-Comb-Auct-Impos, AssadiKS, assad-singla, dutting-combinatorial}. Recently, in a surprising result, Correa, Foncea, Pizarro and Verdugo \cite{CorreaPricing} showed that the reverse direction also holds, establishing an equivalence between finding stopping rules in an optimal stopping problem and designing optimal posted price mechanisms -- for more information on these applications see a survey by Lucier \cite{Lucier-survey}.

The $\nicefrac{1}{2}$-competitive factor guaranteed by the classical prophet inequality for adversarial arrival order has been shown to hold for more general classes of downwards-closed constraints, all the way up to matroids \cite{klein-wein}. For the special case of $k$-uniform matroids, where one can select up to $k$ values, Alaei \cite{Alaei14} showed a $\prn{1 - \frac{1}{\sqrt{k+3}}}$-competitive ratio. This was recently improved for small $k$ via the use of a static threshold \cite{suchi-small-k} and later made tight for all $k$ \cite{willma-adversarial}. Ezra, Feldman, Gravin and Tang \cite{ezra} showed a $0.337$-prophet inequality for matching constraints. Rubinstein \cite{rubin} considered general downwards-closed feasibility constraints and obtained logarithmic approximations. The standard setting has been extended to combinatorial valuation functions \cite{rs, chek-liv}, where one can obtain a constant competitive ratio when maximizing a submodular function but a logarithmic hardness is known for subadditive functions \cite{rs}. Recently, \cite{suchi-non-adaptive-graphic-matroid} studied non-adaptive threshold algorithms for matroid constraints and gave the first constant-factor competitive algorithm for graphic matroids. Qin, Rajagopal, Vardi and Wierman \cite{convex-pi} study the related problem of \emph{convex prophet inequalities}, in which, instead of costs, one sequentially observes random \emph{cost functions} and needs to assign a total mass of $1$ across all functions in an online manner. Our work is a special case of their work, for constant cost functions and identical distributions, which is why we are able to obtain constant competitive ratios compared to the ratio obtained in \cite{convex-pi}, which is polynomial in the number of cost functions $n$.

Esfandiari, Hajiaghayi, Liaghat and Monemizadeh \cite{esf-prophsec} introduced {\em prophet secretary}, in which the arrival order of the random variables is chosen uniformly at random, instead of by an adversary. They gave an adaptive-threshold algorithm that achieves a $\prn{1 - 1/e}$-competitive ratio and showed no algorithm can achieve a factor better than $0.75$. Ehsani, Hajiaghayi, Kesselheim and Singla \cite{EhsaniHKS18} extend this result to matroid constraints. The factor of $1 - 1/e$ was recently beaten, first for the case where the algorithm is allowed to choose the arrival order, called the {\em free order setting} \cite{abolh} and later for random arrival order \cite{azar}. The best currently known ratio is obtained by Correa, Saona and Ziliotto \cite{correa-random}, where they also improve the upper bound to $0.732$. When one can select up to $k$ values, Arnosti and Ma \cite{willma-random} recently gave a surprising and quite beautiful single-threshold algorithm that achieves the best competitive ratio of $1 - e^{-k} \frac{k^k}{k!}$. More general feasibility constraints have also been studied in the random arrival order case, i.e. for matroids \cite{aw-random} and matchings \cite{brubach, tristan}.

Our work is most closely related to the long line of work that considers the case of I.I.D. random variables drawn from a known distribution, which dates back to Hill and Kertz \cite{hill-kertz}. Kertz \cite{kertz} showed that the competitive ratio in the I.I.D. case approaches $\approx 0.745$ as $n$ goes to infinity, via a recursive approach, and conjectured that this is the best bound possible. A simpler proof of this can be found in \cite{saint-mont}. The bound of $\approx 0.745$ was shown to be tight by Correa, Foncea, Hoeksma, Oosterwijk and Vredeveld \cite{correa-iid}. The proofs of both the upper and lower bounds were recently simplified, by \cite{better-tightness} and \cite{renato-iid} respectively. We refer the reader to two excellent surveys \cite{hill-kertz-survey} and \cite{Correa-survey} for more results about prophet inequalities.

Several of these results are described in the context of an {\em online contention resolution scheme (OCRS)}, which is an algorithm used to round a fractional solution of a linear program in an online manner. Originally introduced by Chekuri, Vondr\'{a}k and Zenklusen \cite{crs} for the offline case, Feldman, Svensson and Zenklusen \cite{ocrs} showed the existence of constant-factor approximate OCRSs for several classes of interesting constraints and demonstrated that an $\alpha$-approximate OCRS for a constraint implies an $\alpha$-competitive prophet inequality for the same constraint. This connection was proved to be deeper, as Lee and Singla \cite{lee-singla} used {\em ex-ante prophet inequalities} to design optimal OCRSs for matroids. Recently, in a beautiful series of works, Dughmi \cite{dughmi1, dughmi2} showed that the design of particular (offline) contention resolution schemes is equivalent to another problem in optimal stopping theory, the matroid secretary problem. Whether such connections exists between cost prophet inequalities and OCRSs for upwards-closed constraints is an interesting open question. 

\paragraph{Organization.} Section~\ref{sec:preliminaries} introduces the cost prophet inequality setting and contains relevant definitions as well as important observations. Section~\ref{sec:multi-threshold} characterizes the optimal algorithm and shows that it achieves a tight constant-factor approximation. For the special case of MHR distributions, it shows that this constant is exactly $2$. In Section~\ref{sec:single-threshold}, we focus on single-threshold algorithms and design a fixed threshold that yields a tight $\bigO{\polylog{n}}$-competitive cost prophet inequality. Finally, we conclude with some interesting open problems in Section~\ref{sec:conclusion}.

Due to space constraints and to improve the readability, technical background about the Gamma function as well as some examples and can be found in the Appendices~\ref{app:gamma} and~\ref{app:counterexamples} respectively. All missing proofs can be found in Appendix~\ref{app:lemmas}.

\section{Preliminaries}\label{sec:preliminaries}

In this section we formalize the {\em cost prophet inequality} setting, and define several important quantities. We are given as input $n$ distributions $\cD_1, \dots, \cD_n$ supported on $[0, +\infty)$, and we sequentially observe the realizations of $n$ random {\em costs} $X_1 \sim \cD_1, \dots, X_n \sim \cD_n$. We {\em must} ``stop'' at some point and take the last cost seen. In particular, at any point after observing an $X_i$, we can choose to select or discard it. If we select $X_i$, then the process ends and we receive a cost equal to $X_i$. Otherwise $X_i$ gets {\em discarded} forever and the process continues. An all-knowing prophet, who can see the realizations of all $X_i$'s upfront can always select the minimum realized cost and hence their expected cost is
\[
\mbox{Offline-OPT} = \E\brk{\min_i X_i}.
\]

Given $\cD_1, \dots, \cD_n$, the goal is to design a {\em stopping strategy} that minimizes the expected cost. That is, design an (online) algorithm $\alg$ to decide when to ``stop'' and select the last cost seen, such that the expected cost incurred is minimized, and ideally is comparable to the prophet's cost. Formally, for $\alpha \geq 1$, we say that $\alg$ is $\alpha$-factor approximate/competitive, or achieves an $\alpha$-{\em cost prophet inequality} if
\begin{equation}\label{eq:cpi}
\E\brk{\alg} \leq \alpha \cdot \E\brk{\min_i X_i} = \alpha \cdot \mbox{Offline-OPT}.
\end{equation}

Esfandiari, Hajiaghayi, Liaghat and Monemizadeh \cite{esf-prophsec} observed that, if $X_i$'s are not identically distributed, no algorithm can achieve {\em any} bounded competitive factor. We present a simplified version of this result with $n = 2$ in Appendix~\ref{app:counterexamples}, where the two distributions have support size one and two respectively.

\begin{proposition}\label{prp:negative-arrival}
For the cost prophet inequality problem with adversarial or random order arrival, no algorithm is $\alpha$-factor competitive for any bounded $\alpha$, even when restricted to $n = 2$ and distributions with support size at most two.
\end{proposition}

\paragraph{I.I.D. Setting.} The negative results of the above theorem leads us to consider the case where the $X_i$'s are {\em independent} and {\em identically distributed (I.I.D.)}. Formally, our algorithm sees one-by-one the realizations of $n$ I.I.D. random variables $X_1, \dots, X_n$ drawn from a given distribution $\cD$ with support on $[0, +\infty)$. $\cD$ is defined by its \emph{Cumulative Distribution Function (CDF)} $F : [0, +\infty) \to [0, 1]$, where $F(x) = \Pr_{X \sim \cD}\brk{X \leq x}$, and let $f : [0, +\infty) \rightarrow [0,1]$ denote the \emph{Probability Density Function (PDF)} of $\cD$. 

Given an algorithm $\cA$, let $G_\cA(i)$ denote its expected cost, when it observes $i$ I.I.D. random variables drawn from $\cD$.
Thus, the expected cost of $\cA$ is denoted by $\E\brk{\cA} = G_\cA(n)$. For brevity, we use $\beta_n$ for the remainder of the paper to denote the expected cost of the prophet who can always select the minimum of the $n$ realizations, and $R_\cA(n)$ to denote the competitive ratio of $\cA$ against the prophet's cost $\beta_n$, i.e. $R_\cA(n) = \frac{G_\cA(n)}{\beta_n}$. Whenever the algorithm is clear from context, we drop the subscript and just use $G(n)$ and $R(n)$.

The following observation that characterizes $\beta_n$ is crucial in our analysis.

\begin{observation}\label{obs:betan}
For $n \geq 1$,
\[
\beta_n = \E\brk{\min_{i=1}^n X_i}=\int_0^\infty {\prn{1 - F(s)}^n \dif s}.
\]
\end{observation}

As it turns out, even in the I.I.D. setting, one cannot hope to obtain a bounded competitive factor for {\em all} distributions. The following counterexample due to Lucier \cite{brendan-personal} demonstrates the same.

\begin{theorem}[\cite{brendan-personal}]\label{thm:negative-arrival-equal-revenue}
For the I.I.D. cost prophet inequality problem, no algorithm is $\alpha$-factor competitive for any $\alpha > 0$, even when restricted to $n = 2$.
\end{theorem}
\begin{proof}
Let $n = 2$ and consider the equal-revenue distribution, with support $[1,+\infty)$ and $F(x) = 1 - \nicefrac{1}{x}$. For this distribution, we have
\[
\E[X] = \int^\infty_0 {\prn{1 - F(x)} \dif x} = 1 + \int^\infty_1 {\prn{1 - F(x)} \dif x} = 1 + \int^\infty_1 {\frac{1}{x} \dif x} = + \infty.
\]
In this case, the expected cost of any algorithm is $\E\brk{ALG} = +\infty$, regardless of whether it stops at $X_1$ or at $X_2$. However, the prophet is always able to select the minimum of $X_1$ and $X_2$, which is
\[
\opt = \beta_2 = \int_0^\infty {\prn{1 - F(x)}^2 \dif x} = 1 + \int_1^\infty {\frac{1}{x^2} \dif x} = 2.
\]
Therefore, no algorithm can achieve a finite competitive ratio.
\end{proof}

As it turns out, this strong negative result relies on a certain peculiarity of the distribution related to its {\em hazard rate}. We circumvent this counterexample for a broad class of distributions, which we call {\em Entire distributions} and define below, that includes almost all commonly used distributions such as uniform, exponential, Gaussian, beta and others.

\paragraph{Hazard Rate and Entire Distributions.}
All of our results make heavy use of the {\em hazard (failure) rate} of a distribution.We refer the reader to \cite{ifr-book} for an extensive overview. Intuitively, for discrete distributions, the hazard rate at a point $t$ represents the probability that an event occurs at time $t$, given that the event has not occurred up to time $t$. For continuous distributions, the hazard rate instead quantifies the instantaneous rate of the event's occurrence at time $t$.

\begin{definition}[Hazard Rate]\label{def:hazard-rate}
For a distribution $\cD$ with cumulative distribution function $F$ and probability density function $f$, the {\em hazard rate} of $\cD$ is defined as
\[
h(x) \triangleq \frac{f(x)}{1 - F(x)},
\]
for all $x$ in the support of $\cD$. Furthermore, let $H$ denote the integral of $h$, which we call the {\em cumulative hazard rate} of $\cD$,
\[
H(x) \triangleq \int^x_0 {h(u) \dif u}.
\]
\end{definition}

Next, we express $\beta_n$ in terms of the hazard rate of $\cD$.
\begin{observation}\label{obs:betan-hazard}
\[
\beta_n = \int^\infty_0 {e^{-n H(u)} \dif u}.
\]
\end{observation}
\begin{proof}
Notice that,
\[
H(x) = \int^x_0 {h(u) \dif u} = \int^x_0 {\frac{f(u)}{1 - F(u)} \dif u} = - \int^x_0 {\prn{\ln\prn{1 - F(u)}}' \dif u} = - \ln\prn{1 - F(x)},
\]
which implies that $1 - F(x) = e^{-H(x)}$, and thus, from Observation~\ref{obs:betan}, we have $\beta_n = \int^\infty_0 {e^{-n H(u)} \dif u}.$
\end{proof}

To characterize the class of distributions for which our results hold, we need a generalization of the concept of a Taylor series for a function.

\begin{definition}[Puiseux Series]\label{def:puiseux-series}
We say that a function $f : \R \to \R$ has a Puiseux series expansion if there exist integers $n > 0$ and $i_0 \in \Z$ as well as coefficients $a_1, a_2, \dots$ where $a_1 \neq 0$, such that
\[
f(x) = \sum_{i = i_0}^\infty {a_i x^{\nicefrac{i}{n}}}.
\]
\end{definition}
In other words, Puiseux series are a generalization of Taylor series in that
they allow for fractional exponents in the indeterminate, as long as they
have a bounded denominator. For the remainder of the paper, we will use the simpler form $f(x) = \sum_{i = 1}^\infty {a_i x^{d_i}}$ to denote the Puiseux series of a function $f$, with the understanding that
$d_1, d_2, \dots$ have a bounded denominator. Furthermore, we only focus on
Puiseux series for functions $f : [0, +\infty) \to [0, +\infty)$.

The smallest exponent of the indeterminate in the Puiseux series, $d_1 = \nicefrac{i_0}{n}$ is called the {\em valuation} of $f$ and plays a significant role in our results. The {\em radius of convergence} of a Puiseux series around $0$ is the largest number $r \geq 0$ such that the series converges if $x$ is substituted for a non-zero real number $t \leq r$. A Puiseux series is {\em convergent} at a point $x$ if $x \leq r$.

Now we are ready to define the class of distributions that we focus on in this paper.
\begin{definition}[Entire Distribution]\label{def:entire-distr}
A continuous distribution $\cD$ with support in $[0, + \infty)$ and cumulative hazard rate $H$ is called {\em Entire} if $\E_{X\sim \cD}[X] < +\infty$ and $H$ has a Puiseux series around $0$, $H(x) = \sum_{i = 1}^\infty {a_i x^{d_i}}$, the Puiseux series is not identically zero and is convergent for every point in the support of $\cD$.
\end{definition}

\begin{remark}\label{rmk:support-point}
The definitions and the analysis of our results assumes that the left-most point of the support of $\cD$ is $0$. However, our results easily apply to other distributions with support on $[a,+\infty)$ or $[a,b]$, with $a > 0$, due to the fact that $H(a) = 0$ by definition, and that the Puiseux series of $H$ will be of the form $H(x) = \sum_{i = 1}^\infty {a_i \prn{x - a}^{d_i}}$.
\end{remark}

We note that the class of Entire distributions contains several well-known distributions, including the uniform, exponential, Gaussian, arcsine, beta, gamma, Rayleigh and Weibull distribution, among others. The equal-revenue distribution, however, is not Entire since it has cumulative hazard rate $H(x) = \log{x}$ and the Puiseux series of $\log{x}$ around $x = 1$ is $\sum_{i \geq 0} {\frac{\prn{-1}^{i+1}}{i} \prn{x-1}^i}$, which converges only for $x \in [1,2]$, whereas the support of the distribution is $[1, +\infty)$. This is the reason why Theorem~\ref{thm:negative-arrival-equal-revenue} holds.

Since $h$ is a non-negative function, $H$ is a non-negative and monotonically non-decreasing function. Using this we obtain the following observation.

\begin{observation}\label{obs:a-positive}
Consider an Entire distribution $\cD$ supported on $[0, +\infty)$ with cumulative hazard rate $H(x) = \sum_{i = 1}^\infty {a_i x^{d_i}}$, where $d_1 < d_2 < \dots$. Then, $a_1 > 0$ and $d_1> 0$.
\end{observation}

Distributions with monotonically increasing hazard rate have found a special place within mechanism design literature, originally introduced for the study of revenue maximization. They are known as MHR (or IFR) distributions.

\begin{definition}[Monotone Hazard Rate Distribution]\label{def:mhr}
A distribution $\cD$ is called a {\em Monotone Hazard Rate (MHR)} distribution if and only if the hazard rate function $h$ (Definition \ref{def:hazard-rate}) of $\cD$ is monotonically increasing.
\end{definition}

\paragraph{Gamma function.}

The Gamma ($\Gamma$) function -- which is an extension of the factorial function over the reals -- and its relatives arise in our analysis of the expected cost of the optimal algorithm. For $x > 0$, it is defined as $\Gamma(x) = \int^\infty_0 {t^{x-1} e^{-t} \dif t}$. Of particular interest to us is the {\em lower incomplete Gamma function} $\gamma$, which is defined for $s > 0, x \geq 0$ as $\gamma(s, x) = \int^x_0 {t^{s-1} e^{-t} \dif t}$.

To assist the reader, we include a primer on the Gamma function and its relatives in Appendix~\ref{app:gamma}, along with a few technical lemmas used in our analysis.

\section{Optimal Algorithm: Constant Approximation via Multiple Thresholds}\label{sec:multi-threshold}

In this section, we focus on {\em optimal algorithms} for the cost prophet inequality (CPI) setting, i.e., algorithms that achieve the smallest possible $\alpha$ in \eqref{eq:a-approx}. We show that these algorithms achieve a (distribution-dependent) constant-factor CPI for the class of Entire distributions, and a $2$-factor CPI for Entire MHR distributions. 

We first observe that, just as in the rewards maximization setting, it suffices to focus on threshold-based algorithms to achieve the optimal competitive ratio. A threshold-based algorithm decides thresholds $\tau_1, \dots, \tau_n$ upfront using only the knowledge of the underlying distribution $\cD$, and selects the first $X_i \leq \tau_i$. Since the thresholds do not depend on the realizations of the $X_i$'s, the optimal threshold-based algorithm is an {\em oblivious} algorithm.

\begin{proposition}\label{prp:thresh-is-fine}
For any instance of the cost prophet inequality setting, one can achieve the
optimal competitive ratio with a threshold-based oblivious algorithm.
\end{proposition}

Intuitively, this is because the algorithm's decision in round $i$ is independent of past realizations and only depends on the realization of $X_i$ and the number of remaining random variables, i.e. it is a {\em memoryless} process.

Given Proposition~\ref{prp:thresh-is-fine}, we focus on threshold-based algorithms. Let $\tau_1,\dots,\tau_n$ denote the thresholds of the {\em optimal algorithm}. As it turns out, the optimal thresholds have a very natural interpretation: the algorithm should select the next random variable $X_i$ if and only if its value is smaller than the value it expects to receive by ignoring $X_i$ and continuing the process. That is, the optimal threshold for the next random variable when we have $k$ realizations left to see is exactly the expected cost incurred by an optimal algorithm when its input is $k-1$ I.I.D. random variables.

We then analyze the performance of the optimal-threshold algorithm and show it obtains a constant-factor competitive ratio for the class of Entire distributions.
We identify each Entire distribution by its cumulative hazard rate $H(x)$ and show
that the constant factor depends on the growth rate of $H$\footnote{Recall that
$H$ is non-decreasing, since its derivative, the hazard rate function $h$, is
non-negative.}. Interestingly, the constant-factor is dominated by the {\em valuation} of the Puiseux series of $H$, i.e. the smallest degree in the series. Intuitively, this is because as $n$ gets large enough, the contribution of the other degrees apart from the valuation of $H$ becomes negligible, as our analysis shows, and the behaviour of $H$ is dominated by its valuation. In particular, for a
distribution with $H(x) = \sum_{i = 1}^\infty {a_i \: x^{d_i}}$, where $d_1 < d_2 < \dots$, the precise constant factor we obtain is
\[
\lambda(d_1) = \frac{\prn{1+1/d_1}^{1/d_1}}{\Gamma(1+1/d_1)}.
\]
Perhaps surprisingly, we show that this constant is tight for the distribution with $H(x) = x^{d_1}$.

We view this as both a positive and a negative result; while we can achieve a constant-factor competitive ratio for every fixed distribution, the constant can be arbitrarily large, as $\lim_{d \to 0} \lambda(d) = +\infty$. Since many interesting distributions, with the exception of some pathogenic cases, are Entire distributions, our results imply a (distribution-dependent) constant-factor cost prophet inequality for almost all distributions. An interesting question is whether this behaviour is due to some technicality that exists for infinite support distributions. We answer this negatively in Appendix~\ref{app:bounded-supp}, where, for any $\alpha > 0$, we provide a family of distributions, each supported on $[0,1]$, for which the competitive ratio of the optimal algorithm is exactly $\lambda(\alpha)$.

Finally, we focus on the special case of MHR distributions and show that if $\cD$ is an {\em Entire MHR distribution} with Puiseux series $H(x) = \sum_{i = 1}^\infty {a_i \: x^{d_i}}$, then $d_1 \geq 1$. Since $\lambda$ is decreasing in $d_1$ and $\lambda(1) = 2$, this directly implies a tight $2$-competitive ratio for Entire MHR distributions.

\subsection{Characterizing the Optimal Thresholds}\label{sec:opt-thrs}

In this section we obtain an exact formulation for the optimal thresholds and, using these, design an optimal threshold-based algorithm. In what follows, we use $G(i)$ to denote $G_{\optalg}(i)$ for brevity, where $\optalg$ is an optimal algorithm.

\begin{lemma}\label{lem:opt-thresholds}
For the cost prophet inequality problem with random variables $X_1, X_2, \dots, X_n$, $\tau_n = +\infty$ for every algorithm. For $1 \leq i \leq n-1$, the optimal threshold for the random variable $X_i$ is
\[
\tau_i = G(n-i).
\]
\end{lemma}

Lemma~\ref{lem:opt-thresholds} implies that the following threshold-based algorithm is an optimal algorithm; it achieves the best possible competitive ratio for the cost prophet inequality (CPI) problem.

\begin{algorithm}{\sc{Optimal Threshold Algorithm}}{\cD}\label{alg:optimal}
Set $\tau_n \from + \infty$ and $\tau_{n-1} \from \E_{X \sim \cD}\brk{X}$.
\\  \For{$i \from n-2$ \KwTo $1$}{
    $\tau_i \from F(\tau_{i+1}) \E\brk{X \midd X \leq \tau_{i+1})} + \prn{1 - F(\tau_{i+1})} \tau_{i+1}$.
}
\For{$i \from 1$ \KwTo $n$}{
    Let $z_i$ be the realization of $X_i$.
\\  \If{$z_1, \dots, z_{i-1}$ {\normalfont were not selected and } $z_i \leq \tau_i$}{
        Select $z_i$.
    }
}
\end{algorithm}

\subsection{Constant Factor Competitive Ratio}

In this section, we show Theorem~\ref{thm:hazard-constant}, restated below.

\begin{reptheorem}{thm:hazard-constant}
For the I.I.D. setting under any given non-negative Entire distribution $\cD$, for large enough $n$, there exists a $\lambda(d)$-factor cost prophet inequality, where
\[
\lambda(d) = \frac{\prn{1+1/d}^{1/d}}{\Gamma\prn{1+1/d}},
\]
$d$ is the smallest degree of the Puiseux series of $H$, and $\Gamma(\cdot)$ is the Gamma function.

Moreover, this constant is tight for the distribution with cumulative hazard rate $H(x) = x^d$.
\end{reptheorem}

\subsubsection{Upper Bound}

\begin{theorem}\label{thm:hazard-constant-upper}
Let $\cD$ be an Entire distribution on $[0, +\infty)$ with cumulative hazard rate $H$, which has a Puiseux series $H(x) = \sum_{i = 1}^\infty {a_i x^{d_i}}$ where $d_1 < d_1 < \dots$, and let
\[
\lambda(d_1) = \frac{\prn{1+1/d_1}^{1/d_1}}{\Gamma\prn{1+1/d_1}}.
\]
Then, Algorithm~\ref{alg:optimal} achieves a $\lambda(d_1)$-competitive ratio with respect to $\beta_n$, for large enough $n$.
\end{theorem}
\begin{proof}
By Observation \ref{obs:a-positive} we have that $d_1>0$ and $a_1>0$.
For the competitive ratio of Algorithm~\ref{alg:optimal}, we start by analyzing its expected cost with respect to the cumulative hazard rate $H(x)$.
\begin{lemma}\label{lem:G-recurrence}
The expected cost incurred by Algorithm~\ref{alg:optimal} is
\[
G(n) = \int^{G(n-1)}_0 {e^{-H(u)} \dif u}.
\]
\end{lemma}

Recall that $R(n)$ denotes the competitive ratio of Algorithm~\ref{alg:optimal} for $n$ random
variables, and that our algorithm compares against the prophet who always selects the
minimum value out of all realizations, i.e. $\beta_n$ on expectation.
We want to show that $R(n)$ is upper bounded by a constant for all $n \geq 1$. By
Lemma~\ref{lem:G-recurrence}, we have
\[
R(n) = \frac{G(n)}{\beta_n} = \frac{1}{\beta_n} \int^{G(n-1)}_0 {e^{-H(u)} \dif u} =
\frac{1}{\beta_n} \int^{G(n-1)}_0 {e^{-\sum_{i = 1}^\infty {a_i u^{d_i}}} \dif u}.
\]

Before we proceed, we analyze $\beta_n$.
\begin{lemma}\label{lem:poly-hazard-beta}
For every $n \geq 1$.
\[
\beta_n = \frac{\Gamma\prn{1+1/d_1}}{\prn{a_1 \: n}^{1/d_1}} + \ltlo{\frac{1}{n^{1/d_1}}}.
\]
\end{lemma}

We are now ready to upper bound $R(n)$.

\begin{lemma}\label{lem:ratio-induction}
For every $n \geq 1$, we have
\[
R(n) \leq \frac{\prn{1+1/d_1}^{1/d_1}}{\Gamma(1+1/d_1)}.
\]
\end{lemma}
\begin{proof}
We show that $R(n) \leq \frac{\prn{1+1/d_1}^{1/d_1}}{\Gamma(1+1/d_1)}$ via induction on $n$. For $n = 1$, $R(1) = 1$ and $\frac{\prn{1+1/d_1}^{1/d_1}}{\Gamma(1+1/d_1)} \geq 1$ for all $d_1 > 0$.
For the induction hypothesis, assume $R(k) \leq \frac{\prn{1+1/d_1}^{1/d_1}}{\Gamma(1+1/d_1)}$ for all $k \leq n$, and let
$\lambda(d_1) = \frac{\prn{1+1/d_1}^{1/d_1}}{\Gamma(1+1/d_1)}$ for brevity. For $n+1$ we have
\begin{align}
R(n+1) &= \frac{1}{\beta_{n+1}} \int^{G(n)}_0 {e^{-\sum_{i = 1}^\infty {a_i u^{d_i}}} \dif u} \nonumber \\
&\leq \frac{1}{\beta_{n+1}} \int^{\lambda(d_1) \beta_n}_0 {e^{-\sum_{i = 1}^\infty {a_i u^{d_i}}} \dif u} \nonumber \\
&= \frac{1}{\beta_{n+1}} \int^{\lambda(d_1) \beta_n}_0 {e^{-a_1 u^{d_1}} \cdot \prod_{i = 2}^\infty e^{-a_i u^{d_i}} \dif u} \nonumber \\
\label{eq:pre-sub-ratio} &= \frac{1}{\beta_{n+1}} \int^{\lambda(d_1) \beta_n}_0 {e^{-a_1 u^{d_1}} \cdot \prod_{i = 2}^\infty \sum_{\ell_i \geq 0} \frac{\prn{-a_i u^{d_i}}^{\ell_i}}{\ell_i!} \dif u}.
\end{align}
where the second inequality follows by our induction hypothesis, since $G(n) \leq \lambda(d_1) \beta_n$. Let $x = a_1 u^{d_1} \iff u = \prn{\frac{x}{a_1}}^{1/d_1}$. Also,
\[
\dif x = a_1 d_1 u^{d_1 - 1} \dif u \iff \dif u = \frac{u^{1-d_1}}{a_1 d_1} \dif x = \frac{x^{1/d_1 - 1}}{{a_1}^{1/d_1} d_1} \dif x.
\]
Thus, \eqref{eq:pre-sub-ratio} becomes
\[
R(n+1) \leq \frac{1}{d_1 {a_1}^{1/d_1} \beta_{n+1}} \int^{a_1 \prn{\lambda(d_1) \beta_n}^{d_1}}_0 {e^{-x} x^{1/d_1 - 1} \cdot \prod_{i = 2}^\infty \sum_{\ell_i \geq 0} \frac{\prn{-a_i \prn{\frac{x}{a_1}}^{d_i / d_1}}^{\ell_i}}{\ell_i!} \dif x}
\]
Recall that $\cD$ is an Entire distribution and thus the Puiseux series for
$H$ converges everywhere in the support of $\cD$. Therefore, each term $\sum_{\ell_i \geq 0} \frac{\prn{-a_i \prn{\frac{x}{a_1}}^{d_i / d_1}}^{\ell_i}}{\ell_i!}$ converges to $e^{-a_i \prn{\frac{x}{a_1}}^{d_i / d_1}}$, and thus we can use the distributive law for infinite products \cite{distr-law-infinite} and obtain
\[
\prod_{i = 2}^\infty \sum_{\ell_i \geq 0} \frac{\prn{-a_i \prn{\frac{x}{a_1}}^{d_i / d_1}}^{\ell_i}}{\ell_i!} = \sum_{\ell_2, \ell_3, \ldots \geq 0} \prod_{i = 2}^\infty \frac{\prn{-a_i \prn{\frac{x}{a_1}}^{d_i / d_1}}^{\ell_i}}{\ell_i!}.
\]
Therefore,
\begin{align}
R(n+1) &\leq \frac{1}{d_1 {a_1}^{1/d_1} \beta_{n+1}} \int^{a_1 \prn{\lambda(d_1) \beta_n}^{d_1}}_0 {e^{-x} x^{1/d_1 - 1} \cdot  \sum_{\ell_2, \ell_3, \ldots \geq 0} \prod_{i = 2}^\infty \frac{\prn{-a_i \prn{\frac{x}{a_1}}^{d_i / d_1}}^{\ell_i}}{\ell_i!} \dif x} \nonumber \\
&= \frac{1}{d_1 {a_1}^{1/d_1} \beta_{n+1}} \sum_{\ell_2, \ell_3, \ldots \geq 0} \int^{a_1 \prn{\lambda(d_1) \beta_n}^{d_1}}_0 {e^{-x} x^{1/d_1 - 1} \cdot \prod_{i = 2}^\infty \frac{\prn{-a_i \prn{\frac{x}{a_1}}^{d_i / d_1}}^{\ell_i}}{\ell_i!} \dif x} \nonumber \\
&= \frac{1}{d_1 {a_1}^{1/d_1} \beta_{n+1}} \sum_{\ell_2, \ell_3, \ldots \geq 0} \int^{a_1 \prn{\lambda(d_1) \beta_n}^{d_1}}_0 {e^{-x} x^{1/d_1 + 1/d_1 \sum_{j = 2}^\infty {d_j \ell_j} - 1} \cdot \prod_{i = 2}^\infty \frac{\prn{-a_i {a_1}^{-d_i / d_1}}^{\ell_i}}{\ell_i!} \dif x} \nonumber \\
&= \frac{1}{d_1 {a_1}^{1/d_1} \beta_{n+1}} \sum_{\ell_2, \ell_3, \ldots \geq 0} \prod_{i = 2}^\infty \frac{\prn{-a_i {a_1}^{-d_i / d_1}}^{\ell_i}}{\ell_i!} \int^{a_1 \prn{\lambda(d_1) \beta_n}^{d_1}}_0 {e^{-x} x^{1/d_1 + 1/d_1 \sum_{j = 2}^\infty {d_j \ell_j} - 1} \dif x} \nonumber \\
\label{eq:pre-mu} &= \frac{1}{d_1 {a_1}^{1/d_1} \beta_{n+1}} \sum_{\ell_2, \ell_3, \ldots \geq 0} \prod_{i = 2}^\infty \frac{\prn{-a_i {a_1}^{-d_i / d_1}}^{\ell_i}}{\ell_i!} \gamma\prn{1/d_1 + 1/d_1 \sum_{j = 2}^\infty {d_j \ell_j}, a_1 \prn{\lambda(d_1) \beta_n}^{d_1}},
\end{align}
where the third inequality follows by multiplying together the terms of each sum, the fourth inequality follows by exchanging the order of summation and integration, the fifth inequality follows because the product does not depend on $x$, and the last inequality follows by the definition of $\gamma(s,x)$.

Now, using Fact~\ref{fct:gamma-series}, \eqref{eq:pre-mu} becomes
\begin{align}
R(n+1) &\leq \frac{1}{d_1 {a_1}^{1/d_1} \beta_{n+1}} \sum_{\ell_2, \ell_3, \ldots \geq 0} \prod_{i = 2}^\infty \frac{\prn{-a_i {a_1}^{-d_i / d_1}}^{\ell_i}}{\ell_i!} \prn{a_1 \prn{\lambda(d_1) \beta_n}^{d_1}}^{1/d_1 + 1/d_1 \sum_{j = 2}^\infty {d_j \ell_j}} \nonumber \\
& \qquad \qquad \cdot \sum_{\ell_1 \geq 0} {\frac{\prn{- a_1 \prn{\lambda(d_1) \beta_n}^{d_1}}^{\ell_1}}{\ell_1 ! \: \prn{1/d_1 + 1/d_1 \sum_{j = 2}^\infty {d_j \ell_j} + \ell_1}}} \nonumber \\
&= \frac{{a_1}^{1/d_1} \lambda(d_1) \beta_n}{d_1 {a_1}^{1/d_1} \beta_{n+1}} \sum_{\ell_2, \ell_3, \ldots \geq 0} \prod_{i = 2}^\infty \frac{\prn{-a_i \: \prn{\lambda(d_1) \beta_n}^{d_i} }^{\ell_i}}{\ell_i!} \sum_{\ell_1 \geq 0} {\frac{\prn{- a_1 \prn{\lambda(d_1) \beta_n}^{d_1}}^{\ell_1}}{\ell_1 ! \: \prn{1/d_1 + 1/d_1 \sum_{j = 1}^\infty {d_j \ell_j}}}} \nonumber \\
&= \frac{\lambda(d_1) \beta_n}{\beta_{n+1}} \sum_{\ell_1, \ell_2, \ldots \geq 0} \prod_{i = 2}^\infty \frac{\prn{-a_i \: \prn{\lambda(d_1) \beta_n}^{d_i} }^{\ell_i}}{\ell_i!} {\frac{\prn{- a_1 \prn{\lambda(d_1) \beta_n}^{d_1}}^{\ell_1}}{\ell_1 ! \: \prn{1 + \sum_{j = 1}^\infty {d_j \ell_j}}}} \nonumber \\
\label{eq:integral-done} &= \lambda(d_1) \frac{\beta_n}{\beta_{n+1}} \sum_{\ell_1, \ell_2, \ldots \geq 0} \prod_{i = 1}^\infty \frac{\prn{-a_i \: \prn{\lambda(d_1) \beta_n}^{d_i} }^{\ell_i}}{\ell_i! \prn{1 + \sum_{j = 1}^\infty {d_j \ell_j}}}.
\end{align}

\begin{claim}\label{clm:ratio-upper-bound}
For large enough $n$,
\[
\frac{\beta_n}{\beta_{n+1}} \sum_{\ell_1, \ell_2, \ldots \geq 0} \prod_{i = 1}^\infty \frac{\prn{-a_i \: \prn{\lambda(d_1) \beta_n}^{d_i} }^{\ell_i}}{\ell_i! \prn{1 + \sum_{j = 1}^\infty {d_j \ell_j}}} \leq 1.
\]
\end{claim}

Thus, it follows that Algorithm~\ref{alg:optimal} achieves a
$\frac{\prn{1+1/d_1}^{1/d_1}}{\Gamma(1 + 1/d_1)}$-competitive ratio with respect to $\beta_n$.
\end{proof}
\end{proof}

\begin{remark}
The astute reader might observe that throughout the paper we've assumed that the support of $D$ begins at $0$, which implies that $H(x) = \int_0^x h(u) \dif u$, and thus $H(0) = 0$, which in turn implies that $d_1 > 0$. This is without loss of generality. Specifically, if the support of $D$ begins at $a > 0$, one can ``shift'' it to the origin to find the approximation factor. Formally, we have $H(x)= \int_a^x h(u) du = \int_0^{x-a} h(u+a) du$, and thus $H(a) = 0$. Define $H'(x)=H(x+a)$. We have $H'(0) = 0$ and the approximation factor of the original distribution depends on $d'_1 > 0$. Thus, this dependence is a technicality that does not affect the approximation factor.
\end{remark}

\subsubsection{Lower Bound}
In this section, we show that there exist Entire distributions for which the upper bounds given by $\lambda$ of the previous section is tight. Notice that a cumulative hazard rate $H(x) = x^d$ defines, for $d > 0$, a distribution on $[0, +\infty)$, with CDF $F(x) = 1 - e^{-x^d}$, since $F(0) = 0$ and $\lim_{x \to \infty} F(x) = 1$. For $d = 1$, the resulting distribution is the exponential with rate $1$.

\begin{theorem}\label{thm:hazard-constant-lower}
Consider the distribution $\cD$ for which $H(x) = x^d$ for $d \geq 0$. For any $\eps > 0$, there is no $\prn{\frac{\prn{1+1/d}^{1/d}}{\Gamma\prn{1+1/d}}-\eps}$-competitive cost prophet inequality for the single-item setting and I.I.D. random variables drawn from $\cD$.
\end{theorem}

One can see Theorem~\ref{thm:hazard-constant} as both a positive and a negative result,
since even though for almost all distributions there exists an algorithm that achieves a constant-factor competitive ratio, this constant can be arbitrarily large. 

Now, Theorem~\ref{thm:hazard-constant} follows by Theorems~\ref{thm:hazard-constant-upper} and~\ref{thm:hazard-constant-lower}.

\subsection{Special Case: MHR Distributions}\label{sec:mhr-results}

Even though the constant-factor competitive ratio obtained by Algorithm~\ref{alg:optimal} is distribution-dependent, it turns out that we can show a uniform factor of $2$ when the distributions are MHR. This factor is also tight, and it provides a nice parallel to the standard
$\nicefrac{1}{2}$-competitive prophet inequality in the rewards setting \cite{kren-such, kren-such2, sam-cahn, klein-wein}.

\begin{theorem}\label{thm:mhr-constant-upper}
For every Entire MHR distribution, there exists an algorithm that achieves a $2$-competitive ratio in the cost prophet inequality setting.
\end{theorem}
\begin{proof}
Let $\cD$ be an Entire MHR distribution with cumulative hazard rate $H$ where $H$ has a Puiseux series $H(x) = \sum_{i = 1}^\infty {a_i x^{d_i}}$ and $d_1 < d_2 < \dots$. Notice that since $\cD$ has a monotonically increasing hazard rate,
we have $h'(x) = H''(x) \geq 0$ everywhere in $[0, +\infty)$. Thus,
\[
\prn{\sum_{i = 1}^\infty {a_i x^{d_i}}}'' \geq 0 \iff \prn{\sum_{i = 1}^\infty {a_i \: d_i x^{d_i-1 }}}' \geq 0 \iff \sum_{i = 1}^\infty {a_i \: d_i \: \prn{d_i - 1} x^{d_i - 2}} \geq 0,
\]
for all $x \geq 0$. Recall that, by Observation~\ref{obs:a-positive}, for $H$ to be the cumulative hazard rate of a distribution $\cD$, it must be an increasing function in $x$, and thus $a_1 > 0$. 

Assume towards contradiction that $d_1 < 1$, which implies that the first term of $H$ is negative.
We use this to contradict the fact that $H''(x) \geq 0$ everywhere. In particular, consider a point $y$ where
\begin{align*}
a_1 \: d_1 (1 - d_1) y^{d_1} > \sum_{i = 2}^\infty {a_i \: d_i \: (d_i - 1) y^{d_i}} \iff \\
y^2 \prn{a_1 \: d_1 (1 - d_1) y^{d_1 - 2} - \sum_{i = 2}^\infty {a_i \: d_i \: (d_i - 1) y^{d_i -2}} } &> 0 \iff \\
- y^2 H''(y) > 0 \implies H''(y) < 0.
\end{align*}
Such a point can always be found because, for any choice of $a_1, a_2, \dots$ and $d_1 < d_2 < \dots$, one can pick a small enough $y$ that ensures $a_1 \: d_1 (1 - d_1) y^{d_1}$ dominates the term $\sum_{i = 2}^\infty {a_i \: d_i \: (d_i - 1) y^{d_i}}$.

Therefore, for all Entire MHR distributions, it must be the case that $d_1 \geq 1$. This implies that for every Entire MHR distribution $\cD$, $\lambda(d_1) \leq \lambda(1) = 2$, and thus Algorithm~\ref{alg:optimal} obtains a $2$-factor approximation to the prophet's cost.
\end{proof}

Furthermore, notice that if we consider the distribution with $H(x) = x$, i.e. the exponential distribution, then, as a corollary of Theorem~\ref{thm:hazard-constant} for $d = 1$, we get that the factor of $2$ is tight. The exponential distribution is MHR as it has a constant hazard rate, and hence we obtain the following result.

\begin{theorem}\label{thm:mhr-constant-lower}
For any $\eps > 0$, there exists no $\prn{2-\eps}$-factor cost prophet inequality for the exponential distribution.
\end{theorem}

Now, Theorem~\ref{thm:mhr-constant} follows by Theorems~\ref{thm:mhr-constant-upper} and \ref{thm:mhr-constant-lower}.

\section{Single Threshold Algorithm}\label{sec:single-threshold}

This section is dedicated to proving Theorem~\ref{thm:single-threshold}.
We design an algorithm which sets a fixed threshold $T$ and
selects the first realization that is below $T$. If our algorithm ever reaches $X_n$ and has not selected any value, it is forced to pick the realization of $X_n$ regardless of its cost. Our choice of $T$ is
\[
T = \bigTh{\prn{\frac{\log{n}}{n}}^k},
\]
for an appropriate value of $k$ that depends on the given distribution.

As in Section~\ref{sec:multi-threshold}, we analyze our algorithm's performance for
an Entire distribution with Puiseux series for the cumulative hazard rate $H(x) = \sum_{i = 1}^\infty {a_i \: x^{d_i}}$, where $d_1 < d_2 < \dots$, and obtain a $\bigO{\prn{\log{n}}^{1/d_1}}$-competitive ratio.
We then proceed to show that this ratio is asymptotically tight, as we show that no single threshold algorithm can achieve a competitive ratio better than $\bigOm{\prn{\log{n}}^{1/d}}$ for the distribution with $H(x) = x^d$. As before, our results imply a $\bigO{\polylog{n}}$-factor single-threshold cost prophet inequality for the single-item setting, for almost all distributions.

\subsection{Upper Bound}

\begin{theorem}\label{thm:single-threshold-upper}
Let $\cD$ be an Entire distribution on $[0, +\infty)$ for which the cumulative hazard rate function $H$ has Puiseux series $H(x) = \sum_{i = 1}^\infty {a_i x^{d_i}}$, where $d_1 < d_2 < \dots$. Then, there exists a single threshold $T = T(n, d_1, a_1)$ such that the algorithm that selects the first value $X_i \leq T$ for $i < n$ and $X_n$ otherwise, achieves a $\bigO{\prn{\log{n}}^{1/d_1}}$-competitive ratio compared to $\beta_n$, for large enough $n$.
\end{theorem}
\begin{proof}
We start by analyzing the algorithm's performance for an arbitrary choice of $T$. We have
\begin{align}
\label{eq:single-alg-hazard-first} \E[ALG] &= \prn{1 - \prn{1 - F(T)}^{n-1}} \E\brk{X \midd X \leq T} + \prn{1 - F(T)}^{n-1} \E[X] \\
&= \prn{1 - e^{-(n-1) H(T)}} \int^T_0 {\prn{1 - \frac{F(x)}{F(T)}} \dif x} +
e^{-(n-1) H(T)} \int^\infty_0 {e^{-H(x)} \dif x} \nonumber \\
&= \prn{1 - e^{-(n-1) H(T)}} \int^T_0 {\prn{1 - \frac{1 - e^{-H(x)}}{1 - e^{-H(T)}}} \dif x} +
e^{-(n-1) H(T)} \int^\infty_0 {e^{-H(x)} \dif x} \nonumber \\
&= \frac{1 - e^{-(n-1) H(T)}}{1 - e^{-H(T)}}
\prn{\int^T_0 {e^{-H(x)} \dif x} - T e^{-H(T)}} + e^{-(n-1) H(T)} \int^\infty_0 {e^{-H(x)} \dif x} \iff \nonumber \\
\label{eq:single-alg-hazard} R(n) &= \frac{1}{\beta_n} \prn{\frac{1 - e^{-(n-1) H(T)}}{1 - e^{-H(T)}}
\prn{\int^T_0 {e^{-H(x)} \dif x} - T e^{-H(T)}} + e^{-(n-1) H(T)} \beta_1}.
\end{align}

Notice that,
\[
\int^T_0 {e^{-H(x)} \dif x} - T e^{-H(T)} \leq T \prn{1  - e^{-H(T)}}.
\]
Using the above, \eqref{eq:single-alg-hazard} becomes
\begin{align}
R(n) &\leq \frac{1}{\beta_n} \prn{\frac{1 - e^{-(n-1) H(T)}}{1 - e^{-H(T)}}
T \prn{1  - e^{-H(T)}} + e^{-(n-1) H(T)} \beta_1} \nonumber \\
\label{eq:single-1} &= \frac{1}{\beta_n} \prn{\prn{1 - e^{-(n-1) H(T)}} T + e^{-(n-1) H(T)} \beta_1}.
\end{align}
By Lemma~\ref{lem:poly-hazard-beta}, we know that there exist constants $c_1, c_2 > 0$ such that for large enough $n$, we have
\[
c_1 \frac{\Gamma\prn{1+1/d_1}}{n^{1/d_1}} \leq \beta_n \leq c_2 \frac{\Gamma\prn{1+1/d_1}}{n^{1/d_1}}.
\]
Thus, \eqref{eq:single-1} becomes
\begin{align}
R(n) &\leq \frac{n^{1/d_1}}{c_1 \: \Gamma\prn{1+1/d_1}} \prn{\prn{1 - e^{-(n-1) H(T)}} T + e^{-(n-1) H(T)} \: c_2 \: \Gamma\prn{1+1/d_1}}. \\
\label{eq:single-2} &= \frac{n^{1/d_1}}{c_1 \: \Gamma\prn{1+1/d_1}} \prn{1 - e^{-(n-1) H(T)}} T + e^{-(n-1) H(T)} \frac{c_2}{c_1} n^{1 / d_1}.
\end{align}
Let
\[
T = \prn{\frac{\log\prn{\frac{n}{\log{n}}}}{d_1 \: a_1 \: (n-1)}}^{1/d_1}.
\]
Since $H(T) = \sum_{i = 1}^\infty {a_i T^{d_i}}$, we have
\[
H(T) = a_1 \cdot \frac{\log\prn{\frac{n}{\log{n}}}}{d_1 \: a_1 \: (n-1)} + \sum_{i = 2}^\infty {a_i \prn{\frac{\log\prn{\frac{n}{\log{n}}}}{d_1 \: a_1 \: (n-1)}}^{d_i / d_1}} = \frac{\log\prn{\frac{n}{\log{n}}}}{d_1 \: (n-1)} + \sum_{i = 2}^\infty {a_i \prn{\frac{\log\prn{\frac{n}{\log{n}}}}{d_1 \: a_1 \: (n-1)}}^{d_i / d_1}}.
\]
Since $d_i > d_1$ for all $i \geq 2$, we have that, for large enough $n$,
\[
H(T) \approx a_1 \: T^{d_1},
\]
as $\sum_{i = 2}^\infty {a_i T^{d_i}} = \ltlo{T^{d_1}}$. Thus, \eqref{eq:single-2} becomes
\begin{align*}
R(n) &\leq \frac{n^{1/d_1}}{c_1 \: \Gamma\prn{1+1/d_1}} \prn{1 - e^{-(n-1) \: a_1 \: T^{d_1} }} T + e^{-(n-1) \: a_1 \: T^{d_1}} \frac{c_2}{c_1} n^{1 / d_1} \\
&= \frac{n^{1/d_1}}{c_1 \: \Gamma\prn{1+1/d_1}} \prn{1 - e^{-(n-1) \: a_1 \: \frac{\log\prn{\frac{n}{\log{n}}}}{d_1 \: a_1 \: (n-1)} }} \prn{\frac{\log\prn{\frac{n}{\log{n}}}}{d_1 \: a_1 \: (n-1)}}^{\frac{1}{d_1}} + e^{-(n-1) \: a_1 \: \frac{\log\prn{\frac{n}{\log{n}}}}{d_1 \: a_1 \: (n-1)}} \frac{c_2}{c_1} n^{\frac{1}{d_1}} \\
&= \frac{n^{1/d_1}}{c_1 \: \Gamma\prn{1+1/d_1}} \prn{1 - \prn{\frac{\log{n}}{n}}^{1/d_1}} \prn{\frac{\log\prn{\frac{n}{\log{n}}}}{d_1 \: a_1 \: (n-1)}}^{1/d_1} + \frac{c_2}{c_1} \prn{\frac{\log{n}}{n}}^{1/d_1} \: n^{1 / d_1} \\
&= \frac{1}{{c_1 \: \Gamma\prn{1+1/d_1}} \: \prn{d_1 \: a_1}^{1/d_1}} \prn{\frac{n}{n-1}}^{\frac{1}{d_1}} \prn{1 - \prn{\frac{\log{n}}{n}}^{\frac{1}{d_1}}} \prn{\log\prn{\frac{n}{\log{n}}}}^{\frac{1}{d_1}} + \frac{c_2}{c_1} \prn{\log{n}}^{\frac{1}{d_1}} \\
\end{align*}
However, there exists a constant $c_3 > 0$ such that for large enough $n$,
\[
\prn{1+\frac{1}{n-1}}^{1 + 1/d_1} \prn{1 - \prn{\frac{\log{n}}{n}}^{1/d_1}} \leq c_3,
\]
and also $\prn{\log\prn{\frac{n}{\log{n}}}}^{1/d_1} \leq \prn{\log{n}}^{1/d_1}$. Thus,
\[
R(n) \leq \frac{c_3}{{c_1 \: \Gamma\prn{1+1/d_1}} \: \prn{d_1 \: a_1}^{1/d_1}} \cdot \prn{\log{n}}^{1/d_1} + \frac{c_2}{c_1} \prn{\log{n}}^{1/d_1} = \bigO{\prn{\log{n}}^{1/d_1}}.
\]
\end{proof}

\subsection{Lower Bound}

\begin{theorem}\label{thm:single-threshold-lower}
Consider the distribution $\cD$ for which $H(x) = x^d$ for $d \geq 0$. There is no $\ltlo{\prn{\log{n}}^{1/d}}$-competitive single-threshold cost prophet inequality for the single-item setting and I.I.D. random variables drawn from $\cD$.
\end{theorem}

Theorem~\ref{thm:single-threshold} now follows by Theorems~\ref{thm:single-threshold-upper} and~\ref{thm:single-threshold-lower}.

\section{Conclusion}\label{sec:conclusion}

In this paper, we studied the {\em cost minimization} counterpart of the classical prophet inequality due to Krengel, Sucheston and Garling \cite{kren-such}. The upwards-closed constraint in our setting makes it fundamentally different and more complex. First, the non-I.I.D. case turns out to be intractable in the sense that no finite approximation factor is achievable by any algorithm when the the arrival order is adversarial or random. For the I.I.D. case, we show that if the distribution is Entire, i.e. its cumulative hazard rate $H$ has a convergent Puiseux series, then the best approximation possible is a distribution-dependent constant, which depends on the growth rate of $H$. This constant is at most $2$ for MHR distributions. Beyond Entire distributions, it is unclear if a constant, or any finite, factor approximation is possible, as we present a non-Entire distribution for which no finite factor is possible even when $n = 2$. Furthermore, when restricted to single-threshold algorithms, we show that the best possible approximation is poly-logarithmic, where the power of the logarithm is again a distribution-dependent constant. In all three cases, our results are tight.

Our work opens up a number of interesting questions.
\begin{itemize}
\item The optimal online algorithm has $n$ distinct thresholds, one for each $X_i$, which is at the other extreme compared to the single-threshold algorithms. What if we are allowed to use at most $k$-thresholds for $k > 1$? How does the competitive ratio improve with $k$, starting with the poly-logarithmic factor we show for $k = 1$?

\item Apart from MHR, are there other interesting classes of Entire distributions for which we can get constant-factor approximation for a fixed constant?

\item If one has only sample access to $\cD$, how does the competitive ratio of the optimal algorithm change with the number of samples? This question, with importance in practical applications when the distributions are not fully known, has been studied extensively in the rewards maximization setting \cite{sample-pi-1, sample-pi-2, sample-pi-3}.

\item An interesting non-I.I.D. setting for which our impossibility results do not apply is the {\em free order} setting in which the distributions can differ but the algorithm can select the order in which it sees the realizations. One cannot hope to do better than in the I.I.D. setting, but is a (distribution-dependent) constant-factor competitive ratio for Entire distributions possible?
\end{itemize}


\bibliographystyle{alpha}
\bibliography{references}


\appendix

\section*{Appendix}

\section{Background on the Gamma Function}\label{app:gamma}

The Gamma function $\Gamma(x)$ extends the factorial function to complex numbers. In particular,
\[
\Gamma(n+1) = n!
\]
for every $n \in \N$.

Here we give a brief and incomplete primer on the Gamma function, to assist the reader. However, for a more extensive treatment along with many folklore results about the function, see \cite{gamma-book}.

\begin{definition}[Gamma ($\Gamma$) Function]
For every $x > 0$, the {\em Gamma function} is defined as
\[
\Gamma(x) = \int^\infty_0 {t^{x-1} e^{-t} \dif t}.
\]
\end{definition}

Like the factorial function, the Gamma function also satisfies the following recurrence
\[
\Gamma(x+1) = x \Gamma(x).
\]

The following fact is closely related to Stirling's approximation for the Gamma function and is due to \cite[Eq.~5.11.E7]{gamma-facts}.
\begin{fact}\label{fct:big-gamma}
For $a > 0$ and $b \in \R$, we have
\[
\Gamma(a+b) \leq \sqrt{2 \pi} \prn{\frac{a}{e}}^{a} \cdot a^b.
\]
\end{fact}

Of particular use to us are the following special functions that are related to the Gamma function.
\begin{definition}[Upper ($\Gamma(\cdot , \cdot)$ and Lower $\gamma(\cdot , \cdot)$ Incomplete Gamma Functions]
For every $s > 0, x \geq 0$, the {\em Upper Incomplete Gamma function} is defined as
\[
\Gamma(s, x) = \int^\infty_x {t^{s-1} e^{-t} \dif t},
\]
whereas the {\em Lower Incomplete Gamma function} is defined as
\[
\gamma(s, x) = \int^x_0 {t^{s-1} e^{-t} \dif t}.
\]
\end{definition}

For every $s > 0, x \geq 0$, we have
\[
\Gamma(s,x) + \gamma(s,x) = \Gamma(s).
\]

Next, we describe a few known results about the lower incomplete Gamma function that we use throughout the paper.

\begin{fact}\label{fct:gamma-series}
For the lower incomplete Gamma function $\gamma(s,x)$ with $s, x > 0$, we have
\[
\gamma(s, x) = x^s \sum_{k = 0}^\infty {\frac{\prn{-x}^k}{k! \: \prn{s+k}}}.
\]
\end{fact}
\begin{proof}
By the definition of the lower incomplete Gamma function, we have
\[
\gamma(s, x) = \int^x_0 {t^{s-1} e^{-t} \dif t} = \int^x_0 {\sum_{k = 0}^\infty {\prn{-1}^k \frac{t^{s+k-1}}{k!}}} = \sum_{k = 0}^\infty {\prn{-1}^k \frac{x^{s+k}}{k! \: \prn{s+k}}} = x^s \sum_{k = 0}^\infty {\frac{\prn{-x}^k}{k! \: \prn{s+k}}}.
\]
\end{proof}

The following fact follows easily via Fact~\ref{fct:gamma-series}.

\begin{fact}\label{fct:small-gamma-approx}
We have that, as $x \to 0$,
\[
\frac{\gamma(s,x)}{x^s} \to s^{-1}.
\]
\end{fact}

The following claim is due to Qi and Mei \cite{qi-mei-gamma}.
\begin{claim}\label{clm:small-gamma-approximation-2}[See 3.1 in \cite{qi-mei-gamma}]
For small enough $x$, we have
\[
\gamma(s, x) \leq s^{-1} \: x^{s-1} \: e^{-x}.
\]
\end{claim}

\section{Counterexamples}\label{app:counterexamples}

\subsection{Single-Threshold Counterexample}

For the cost-prophet inequality setting, a natural approach that is seemingly intuitive is to set a single threshold $T$ close to $\beta_n = \E[\min_i X_i]$ since, if $n$ is large enough, with good probability there will be a realization below the threshold and, this way, one would achieve a very good competitive ratio.

We present an example that shows why this natural intuition fails.

\begin{example}\label{exm:intuition-fail}
Consider the exponential distribution, for which $F(x) = 1 - e^{-x}$, $f(x) = e^{-x}$, $H(x) = x$, $E[X] = 1$ and
\[
\beta_n = \int^\infty_0 {e^{-n x} \dif x} = \frac{1}{n}.
\]
In our attempt to achieve a constant competitive ratio, we set a threshold $T = \frac{c}{n}$ for some constant $c > 0$. If there exists a realization of $X_1, \dots, X_{n-1}$ that is below $T$, then we would select it; otherwise we are forced to select $X_n$ and obtain a cost equal to $\E[X]$.

The probability that there exists a realization of $X_1, \dots, X_{n-1}$ that is below $T$ is $1 - \prn{1 - F(T)}^{n-1}$. Thus, the expected cost of our algorithm is
\begin{align*}
\E[ALG_n] &= \prn{1 - \prn{1 - F(T)}^{n-1}} \E\brk{X \midd X \leq T} + \prn{1 - F(T)}^{n-1} E[X] \\
&= \prn{1 - e^{-(n-1) T}} \E\brk{X \midd X \leq T} + e^{-(n-1) T} \cdot 1 \\
&= \prn{1 - e^{-c \frac{n-1}{n}}} \E\brk{X \midd X \leq \nicefrac{c}{n}} + e^{-c \frac{n-1}{n}} \\
&= \prn{1 - e^{-c \frac{n-1}{n}}} \frac{\int_0^{\frac{c}{n}} {x f(x) \dif x}}{1 - e^{-\nicefrac{c}{n}}} + e^{-c \frac{n-1}{n}} \\
&= \prn{1 - e^{-c \frac{n-1}{n}}} \frac{\int_0^{\frac{c}{n}} {x e^{-x} \dif x}}{1 - e^{-\nicefrac{c}{n}}} + e^{-c \frac{n-1}{n}} \\
&= \prn{1 - e^{-c \frac{n-1}{n}}} \frac{1 - e^{-\nicefrac{c}{n}} - \frac{c}{n} \; e^{-\nicefrac{c}{n}}}{1 - e^{-\nicefrac{c}{n}}} + e^{-c \frac{n-1}{n}} \\
&= \prn{1 - e^{-c \frac{n-1}{n}}} \prn{1 - \frac{c}{n} \cdot \frac{e^{-\nicefrac{c}{n}}}{1 - e^{-\nicefrac{c}{n}}}} + e^{-c \frac{n-1}{n}}.
\end{align*}
Thus, the competitive ratio is
\[
R(n) = \frac{\E[ALG_n]}{\beta_n} = n \prn{\prn{1 - e^{-c \frac{n-1}{n}}} \prn{1 - \frac{c}{n} \cdot \frac{e^{-\nicefrac{c}{n}}}{1 - e^{-\nicefrac{c}{n}}}} + e^{-c \frac{n-1}{n}}}.
\]
Notice that, as $n \to + \infty$, we have
\[
\lim_{n \to +\infty} {n \prn{1 - e^{-c \frac{n-1}{n}}} \prn{1 - \frac{c}{n} \cdot \frac{e^{-\nicefrac{c}{n}}}{1 - e^{-\nicefrac{c}{n}}}}} = \frac{c \: e^{-c} \prn{e^{c} - 1}}{2},
\]
but
\[
\lim_{n \to +\infty} {n \: e^{-c \frac{n-1}{n}}} = + \infty,
\]
and thus the competitive ratio of this algorithm is infinite.
\end{example}

\subsection{Non-I.I.D.: Adversarial or Random Order}

\begin{repproposition}{prp:negative-arrival}
For the cost prophet inequality problem with adversarial or random order arrival, no algorithm is $\alpha$-factor competitive for any bounded $\alpha$, even when restricted to $n = 2$ and distributions with support size at most two.
\end{repproposition}

The proposition follows from the following example.

\begin{example}\label{exm:negative-arrival}
Let $n = 2$ and consider the following random variables:
\begin{gather*}
X_1 = 1 \quad \text{w.p. } 1, \qquad X_2 = \begin{cases}
0 &\quad \text{w.p. } 1 - 1 / L \\
L &\quad \text{w.p. } 1 / L
\end{cases},
\end{gather*}
for an arbitrarily large number $L>0$. If the arrival order of $X_1$ and $X_2$ is adversarial, an adversary can force every algorithm to see $X_1$ before $X_2$. In this case, every algorithm receives an expected value of $\E\brk{ALG} = 1$, regardless of whether it stops at $X_1$ or at $X_2$. However, the prophet will select $X_1=1$ whenever $X_2=L$, and $X_2=0$ otherwise. Thus the prophet's expected cost is 
\[
\opt = 0 \cdot (1 - 1 / L) + 1 \cdot 1 / L = 1 / L,
\]
which implies an $L$-competitive factor.
For random arrival order, notice that with probability $1 / 2$, the algorithm sees $X_1$ before $X_2$ and thus our previous analysis holds. Therefore, $\E\brk{ALG} \ge 1 / 2$, which implies a competitive factor at least $L / 2$.

Since $L$ is arbitrary large, the competitive factor can be made arbitrarily large. 
\end{example}

\subsection{Bounded Support Distributions}\label{app:bounded-supp}

\begin{observation}
For any $\alpha > 0$, there exists a distribution $\cD_{\alpha}$, supported on $[0,1]$ such that for the I.I.D. cost prophet inequality setting with random variables drawn from $\cD_{\alpha}$
\begin{enumerate}
    \item there exists an $\alpha$-competitive cost prophet inequality, and
    \item there does not exist an $\prn{\alpha - \eps}$-competitive cost prophet inequality for any constant $\eps > 0$.
\end{enumerate}
\end{observation}
\begin{proof}
Consider the Beta distribution, which is supported on $[0,1]$ and is parameterized by $\alpha > 0$, and for which $F_{\alpha}(x) = x^{\alpha}$. For this distribution, we have
\[
H_{\alpha}(x) = -\log{\prn{1 - F_{\alpha}(x)}} = \log{\prn{\frac{1}{1 - x^{\alpha}}}}.
\]
The Puiseux series of $H_{\alpha}$ around $x = 0$ is
\[
H_{\alpha}(x) = \sum_{k \geq 1} {\frac{x^{k \alpha}}{k}},
\]
which converges for $x \in [0,1)$\footnote{The Puiseux series and $H$ are equal also for $x \to 1$, since both diverge to $+\infty$.}. Thus, we observe that, for this distribution, $d_1 = \alpha$, and from Theorem~\ref{thm:hazard-constant}, we know that there exists a tight $\frac{\prn{1 + 1/\alpha}^{1/\alpha}}{\Gamma\prn{1+1/\alpha}}$-cost prophet inequality.
\end{proof}

\section{Missing Proofs}\label{app:lemmas}

\subsection{Proof of Observation \ref{obs:betan}}

\begin{repobservation}{obs:betan}
For $n \geq 1$,
\[
\beta_n = \E\brk{\min_{i=1}^n X_i}=\int_0^\infty {\prn{1 - F(s)}^n \dif s}.
\]
\end{repobservation}
\begin{proof}
Let $Y_n = \min_{i = 1}^n {X_i}$, then the CDF $F_{Y_n}$ of $Y_n$ is
\[
F_{Y_n}(x) = \Pr\brk{Y_n \leq x} = 1 - \Pr\brk{Y_n > x} = 1 - \prod_{i = 1}^n {\Pr\brk{X_i > x}} = 1 - \prn{1 - F(x)}^n, \hfill \forall x \in [0, +\infty).
\]
Recall that for a random variable $X$, we have
\[
\E[X] = \int_0^\infty {x f_X(x) \dif x} = \int_0^\infty {f_X(x) \int^x_0 \dif t \dif x}.
\] By changing the order of integration, we obtain
\[
\E[X] = \int_0^\infty {\int^x_0 f_X(x) \dif x \dif t} = \int_0^\infty {\Pr[X \geq t] \dif t} = \int_0^\infty {\prn{1 - F_X(t)} \dif t}.
\] Using this, we get that the expected cost of the prophet (offline optimum), denoted by $\beta_n$ is,
\[ 
\beta_n = \E[Y_n] = \int_0^\infty {\prn{1 - F_{Y_n}(s)} \dif s}
= \int_0^\infty {\prn{1 - \prn{1 -\prn{1 - F(s)}^n}} \dif s}
= \int_0^\infty {\prn{1 - F(s)}^n \dif s}.
\]
\end{proof}

\subsection{Proof of Observation~\ref{obs:a-positive}}

\begin{repobservation}{obs:a-positive}
Consider an Entire distribution $\cD$ supported on $[0, +\infty)$ with cumulative hazard rate $H(x) = \sum_{i = 1}^\infty {a_i x^{d_i}}$, where $d_1 < d_2 < \dots$. Then, $a_1 > 0$ and $d_1> 0$.
\end{repobservation}
\begin{proof}
Once can easily see that $a_1 > 0$ since $H$ is non-negative. Note that, for any choice of $a_1, a_2, \dots$ and $d_1 < d_2 < \dots$, since the Puiseux series of $H$ is convergent for every $x$ in the support of $\cD$, there exists a small enough $x_* \in [0, 1)$ such that,
\[
\abs{a_1 x^{d_1}_*} > \sum_{i = 2}^\infty \abs{a_i x^{d_i}_*}.
\] 
Thus, if $a_1 < 0$, we have $H(x_*) < 0$, a contradiction.

Next we show that $d_1 \geq 0$. Consider the derivative of $H$, namely $h(x) = \sum_{i=1}^\infty a_i d_i x^{d_i -1}$. Again, given that fact that $d_1 < d_i$ for all $i \geq 2$, there exists $y_*$ such that
\[
\abs{a_1 d_1 y^{d_1-1}_*} > \sum_{i = 2}^\infty \abs{a_i d_i y^{d_i-1}_*}.
\]
Thus since $a_1 > 0$, we have $a_1 \cdot d_1 < 0$ which implies $h(y_*) < 0$, a contradiction to $h$ being non-negative.
\end{proof}

\subsection{Proof of Proposition~\ref{prp:thresh-is-fine}}

\begin{repproposition}{prp:thresh-is-fine}
For any instance of the cost prophet inequality setting, one can achieve the
optimal competitive ratio with a threshold-based oblivious algorithm.
\end{repproposition}
\begin{sketch}
Since every algorithm has to select a value, if an algorithm observes the realization of $X_n$, it is forced to select it. When an algorithm sees $X_{n-1}$, it has to decide whether to select it or not. Whatever the decision process of the algorithm, let $p^{\cA}\prn{r \midd X_{n-1} = z}$ be the probability that algorithm $\cA$ selects the realization of $X_i$, given $X_i$. Then, the expected cost of $\cA$ is
\[
\sum_{z \geq 0} z \: p^{\cA}\prn{r \midd X_{n-1} = z} + \prn{1 - \sum_{z \geq 0} p^{\cA}\prn{r \midd X_{n-1} = z}} \E[X_n].
\]
For a fixed choice of $L = \sum_{z \geq 0} {p^{\cA}\prn{r \midd X_{n-1} = z}}$, to maximize this quantity, $\cA$ will greedily assign all the probability mass of $L$ to the lowest values $z$. Thus, the only choice $\cA$ has to make is $L$ itself, which is equal to $\Pr\brk{X_{n-1} \leq F^{-1}(L)}$. Therefore, every choice of $L$ implies a threshold, namely $F^{-1}(L)$.

Finally, for the remaining random variables, the observation holds via induction, since the random variables are I.I.D.
\end{sketch}

\subsection{Proof of Lemma~\ref{lem:opt-thresholds}}

\begin{replemma}{lem:opt-thresholds}
For the cost prophet inequality problem with random variables $X_1, X_2, \dots, X_n$, $\tau_n = +\infty$ for every algorithm. For $1 \leq i \leq n-1$, the optimal threshold for the random variable $X_i$ is
\[
\tau_i = G(n-i).
\]
\end{replemma}
\begin{proof}
The lemma follows by backwards induction on $n$.

\noindent{\em Base case.} Since we are forced to select a single value, if the algorithm ever observes $X_n$, it must select its realization. This is equivalent to $\tau_n = + \infty$. It then follows that $G(1) = \E_{X \sim \cD}[X]$.

\noindent{\em Induction.} Consider the $i$-th step, where $i < n$. For our induction hypothesis, assume that $\tau_j = G(n-j)$ for all $i < j < n$. Conditioned to the fact that the algorithm has reached the $i$-th step, the expected cost of the optimal algorithm is $G(n - i + 1)$; i.e. the cost that the optimal algorithms expects to receive from the remaining $n - i + 1$ variables. Since $\tau_i$ is the optimal threshold for $X_i$, we obtain the following recurrence for $G(n - i + 1)$.
\begin{equation}\label{eq:recurrence}
G(n-i+1) = F(\tau_i) \E\brk{X \midd X \leq \tau_i} + \prn{1 - F(\tau_i)} G(n-i).
\end{equation}
This is because with probability $F(\tau_i)$ we select $X_i$ and therefore receive cost $\E\brk{X \midd X \leq \tau_i}$, and with probability $1 - F(\tau_i)$, we ignore $X_i$ and we receive cost equal to the expected value of the optimal algorithm on $X_{i+1}, \dots, X_n$, i.e. $G(n-i)$. Thus, it suffices to show that setting $\tau_i = G(n - i)$ minimizes $G(n - i + 1)$.

We rearrange \eqref{eq:recurrence} and obtain
\begin{align*}
G(n-i+1) &= F(\tau_i) \E\brk{X \midd X \leq \tau_i} + \prn{1 - F(\tau_i)} G(n-i) \\
&= F(\tau_i) \cdot \frac{\int^{\tau_i}_0 {u f(u) \dif u}}{F(\tau_i)} + \prn{1 - F(\tau_i)} G(n-i) \\
&= \int^{\tau_i}_0 {u f(u) \dif u} + \prn{1 - F(\tau_i)} G(n-i) \\
&= \int^{\tau_i}_0 {u \prn{F(u)}' \dif u} + \prn{1 - F(\tau_i)} G(n-i) \\
&= \brk{u F(u)}^{\tau_i}_0 - \int^{\tau_i}_0 {F(u) \dif u} + \prn{1 - F(\tau_i)} G(n-i) \\
&= \tau_i F(\tau_i) - \int^{\tau_i}_0 {F(u) \dif u} + \prn{1 - F(\tau_i)} G(n-i). \\
\end{align*}
where the second equality follows by the definition of $\E\brk{X \midd X \leq \tau_i}$ and the second-to-last equality follows via integration by parts.

We will show that the optimal threshold at the $i$-th step is
\[
\tau_i = G(n-i).
\]
In other words, we will show that
\begin{align}
G(n-i) F(G(n-i)) - \int^{G(n-i)}_0 {F(u) \dif u} + \prn{1 - F(G(n-i))} G(n-i) & \nonumber \\
\label{eq:threshold-ineq} \leq \tau_i F(\tau_i) - \int^{\tau_i}_0 {F(u) \dif u} + \prn{1 - F(\tau_i)} G(n-i)&,
\end{align}
for any $\tau_i \neq G(n-i)$. Rearranging \eqref{eq:threshold-ineq}, we get
\begin{align}
& G(n-i) F(G(n-i)) - \int^{G(n-i)}_0 {F(u) \dif u} + \prn{1 - F(G(n-i))} G(n-i) \nonumber \\
& \qquad \qquad \qquad \leq \tau_i F(\tau_i) - \int^{\tau_i}_0 {F(u) \dif u} + \prn{1 - F(\tau_i)} G(n-i) \iff \nonumber \\
& G(n-i) - \int^{G(n-i)}_0 {F(u) \dif u} \leq \tau_i F(\tau_i) - \int^{\tau_i}_0 {F(u) \dif u} + G(n-i) - F(\tau_i) G(n-i) \iff \nonumber \\
& F(\tau_i) \prn{G(n-i) - \tau_i} \leq \int^{G(n-i)}_0 {F(u) \dif u} - \int^{\tau_i}_0 {F(u) \dif u} \iff \nonumber \\
\label{eq:threshold-ineq2} & F(\tau_i) \prn{G(n-i) - \tau_i} \leq \int^{G(n-i)}_{\tau_i} {F(u) \dif u}.
\end{align}
We distinguish between two cases: $\tau_i < G(n-i)$ and $\tau_i > G(n-i)$.
In the case where $\tau_i < G(n-i)$, \eqref{eq:threshold-ineq2} becomes
\[
F(\tau_i) \leq \frac{\int^{G(n-i)}_{\tau_i} {F(u) \dif u}}{G(n-i) - \tau_i},
\]
which is true by the mean value theorem, since $F$ is increasing and
$\tau_i < G(n-i)$. Similarly, in the case where $\tau_i > G(n-i)$, \eqref{eq:threshold-ineq2} becomes
\[
F(\tau_i) \geq \frac{\int^{G(n-i)}_{\tau_i} {F(u) \dif u}}{G(n-i) - \tau_i} = \frac{\int^{\tau_i}_{G(n-i)} {F(u) \dif u}}{\tau_i - G(n-i)},
\]
which is again true by the mean value theorem, since $F$ is increasing and
$\tau_i > G(n-i)$.

We conclude that the optimal threshold for $X_i$ is
\[
\tau_i = G(n-i).
\]
\end{proof}

\subsection{Proof of Lemma~\ref{lem:G-recurrence}}

\begin{replemma}{lem:G-recurrence}
The expected cost incurred by Algorithm~\ref{alg:optimal} is
\[
G(n) = \int^{G(n-1)}_0 {e^{-H(u)} \dif u}.
\]
\end{replemma}
\begin{proof}
Recall that $G(n)$ satisfies the recurrence relation in \eqref{eq:recurrence}
\[
G(n) = F(\tau_1) \E\brk{X \midd X \leq \tau_1} + \prn{1 - F(\tau_1)} G(n-1).
\]
Substituting the optimal thresholds from Lemma~\ref{lem:opt-thresholds} into the recurrence above, we obtain
\begin{align*}
G(n) &= F(G(n-1)) \E\brk{X \midd X \leq G(n-1)} + \prn{1 - F(G(n-1))} G(n-1) \\
&= F(G(n-1)) \frac{\int^{G(n-1)}_0 {u f(u) \dif u}}{F(G(n-1))} + \prn{1 - F(G(n-1))} G(n-1) \\
&= \int^{G(n-1)}_0 {u f(u) \dif u} + \prn{1 - F(G(n-1))} G(n-1) \\
&= \brk{u F(u) \dif u}^{G(n-1)}_0 - \int^{G(n-1)}_0 {F(u) \dif u} + \prn{1 - F(G(n-1))} G(n-1) \\
&= G(n-1) F(G(n-1)) - \int^{G(n-1)}_0 {F(u) \dif u} + G(n-1) - G(n-1) F(G(n-1)) \\
&= \int^{G(n-1)}_0 {\prn{1 - F(u)} \dif u}.
\end{align*}
Next, recall that $H(x) = - \log\prn{1 - F(x)}$, and thus we obtain
\[
G(n) = \int^{G(n-1)}_0 {e^{-H(u)} \dif u}.
\]
\end{proof}

\subsection{Proof of Lemma~\ref{lem:poly-hazard-beta}}

\begin{replemma}{lem:poly-hazard-beta}
For every $n \geq 1$.
\[
\beta_n = \frac{\Gamma\prn{1+1/d_1}}{\prn{a_1 \: n}^{1/d_1}} + \ltlo{\frac{1}{n^{1/d_1}}}.
\]
\end{replemma}
\begin{proof}
\begin{align}
\beta_n &= \int^\infty_0 {e^{-n H(u)} \dif u} = \int^\infty_0 {e^{-n \: \sum_{i = 1}^\infty {a_i u^{d_i}}} \dif u} = \int^\infty_0 {e^{-n \: a_1 u^{d_1}} \cdot e^{-n \: \sum_{i = 2}^\infty {a_i u^{d_i}}} \dif u} \nonumber \\
\label{eq:betas-pre-transformation} &= \int^\infty_0 {e^{-n \: a_1 u^{d_1}} \cdot \prod_{i = 2}^\infty {e^{-n \: a_i u^{d_i}}} \dif u} = \int^\infty_0 {e^{-n \: a_1 u^{d_1}} \cdot \prod_{i = 2}^\infty {\sum_{\ell_i \geq 0} {\frac{\prn{-n \: a_i u^{d_i}}^{\ell_i}}{\ell_i !} }} \dif u}.
\end{align}
Let $x = n \: a_1 u^{d_1} \iff u = \prn{\frac{x}{n \: a_1}}^{1/d_1}$. Then,
\[
\dif x = n \: a_1 d_1 u^{d_1 - 1} \dif u \iff \dif u = \frac{u^{1-d_1}}{n \: a_1 \: d_1} \dif x = \frac{x^{1/d_1 - 1}}{{n \: a_1}^{1/d_1} \: d_1} \dif x,
\]
and \eqref{eq:betas-pre-transformation} becomes
\[
\beta_n = \frac{1}{{n \: a_1}^{1/d_1} \: d_1} \int^\infty_0 {e^{-x} x^{1/d_1 - 1} \cdot \prod_{i = 2}^\infty {\sum_{\ell_i \geq 0} {\frac{\prn{-n \: a_i \prn{\frac{x}{n \: a_1}}^{d_i / d_1}}^{\ell_i}}{\ell_i !} }} \dif x}.
\]
As in the proof of Lemma~\ref{lem:ratio-induction}, each term $\sum_{\ell_i \geq 0} {\frac{\prn{-n \: a_i \prn{\frac{x}{n \: a_1}}^{d_i / d_1}}^{\ell_i}}{\ell_i !}}$ converges to $e^{-n \: a_i \prn{\frac{x}{n \: a_1}}^{d_i / d_1}}$, and thus we have
\[
\prod_{i = 2}^\infty {\sum_{\ell_i \geq 0} {\frac{\prn{-n \: a_i \prn{\frac{x}{n \: a_1}}^{d_i / d_1}}^{\ell_i}}{\ell_i !} }} = \sum_{\ell_2, \ell_3, \ldots \geq 0} {\prod_{i = 2}^\infty {\frac{\prn{-n \: a_i \prn{\frac{x}{n \: a_1}}^{d_i / d_1}}^{\ell_i}}{\ell_i !} }}
\]
\begin{align}
\beta_n &= \frac{1}{{n \: a_1}^{1/d_1} \: d_1} \int^\infty_0 {e^{-x} x^{1/d_1 - 1} \cdot \sum_{\ell_2, \ell_3, \ldots \geq 0} {\prod_{i = 2}^\infty {\frac{\prn{-n \: a_i \prn{\frac{x}{n \: a_1}}^{d_i / d_1}}^{\ell_i}}{\ell_i !} }} \dif x} \nonumber \\
&= \frac{1}{{n \: a_1}^{1/d_1} \: d_1} \sum_{\ell_2, \ell_3, \ldots \geq 0} {\int^\infty_0 {e^{-x} x^{1/d_1 - 1} \cdot \prod_{i = 2}^\infty {\frac{\prn{-n \: a_i \prn{\frac{x}{n \: a_1}}^{d_i / d_1}}^{\ell_i}}{\ell_i !} } \dif x}} \nonumber \\
&= \frac{1}{{n \: a_1}^{1/d_1} \: d_1} \sum_{\ell_2, \ell_3, \ldots \geq 0} {\int^\infty_0 {e^{-x} x^{1/d_1 + 1/d_1 \sum_{j = 2}^\infty {d_j \ell_j} - 1} \cdot \prod_{i = 2}^\infty {\frac{\prn{-n \: a_i \prn{n \: a_1}^{- d_i / d_1}}^{\ell_i}}{\ell_i !} } \dif x}} \nonumber \\
&= \frac{1}{{n \: a_1}^{1/d_1} \: d_1} \sum_{\ell_2, \ell_3, \ldots \geq 0} {\prod_{i = 2}^\infty {\frac{\prn{-n \: a_i \prn{n \: a_1}^{- d_i / d_1}}^{\ell_i}}{\ell_i !} } \cdot \int^\infty_0 {e^{-x} x^{1/d_1 + 1/d_1 \sum_{j = 2}^\infty {d_j \ell_j} - 1} \dif x}} \nonumber \\
\label{eq:betas-gamma-appearance} &= \frac{1}{{n \: a_1}^{1/d_1} \: d_1} \sum_{\ell_2, \ell_3, \ldots \geq 0} {\prod_{i = 2}^\infty {\frac{\prn{-n \: a_i \prn{n \: a_1}^{- d_i / d_1}}^{\ell_i}}{\ell_i !} } \cdot \Gamma\prn{1/d_1 + 1/d_1 \sum_{j = 2}^\infty {d_j \ell_j}} }  \\
&= \frac{\Gamma\prn{1 / d_1}}{d_1 \prn{n \: a_1}^{1/d_1}} + \nonumber \\
&\qquad + \frac{1}{d_1 \prn{n \: a_1}^{1/d_1}} \sum_\stack{\ell_2, \ell_3, \ldots \geq 0}{\ell_2, \ell_3 \dots \neq (0, 0, \dots)} \prod_{i = 2}^\infty \frac{\prn{-n \: a_i \prn{n \: a_1}^{-d_i / d_1}}^{\ell_i}}{\ell_i!} \Gamma\prn{1/d_1 + 1/d_1 \sum_{j = 2}^\infty {d_j \ell_j}} \nonumber \\
&= \frac{\Gamma\prn{1 + 1 / d_1}}{\prn{n \: a_1}^{1/d_1}} + \frac{1}{d_1 \prn{n \: a_1}^{1/d_1}} \sum_\stack{\ell_2, \ell_3, \ldots \geq 0}{\ell_2, \ell_3, \dots \neq (0, 0, \dots)} n^{\sum_{j = 2}^\infty {\ell_j \prn{1 - d_j / d_1}}} \nonumber \\
\label{eq:last-eq-of-betas} &\qquad \qquad \cdot \prod_{i = 2}^\infty \frac{\prn{- a_i a_1^{-d_i / d_1}}^{\ell_i}}{\ell_i!} \Gamma\prn{1/d_1 + 1/d_1 \sum_{j = 2}^\infty {d_j \ell_j}}.
\end{align}
where \eqref{eq:betas-gamma-appearance} follows by the definition of the Gamma function.

Notice that, since $d_i > d_1$ for all $i \geq 2$, we have $n^{1 - d_i / d_1} = \ltlo{1}$ for all $i \geq 2$. Also, in the exponent of $n$ in the second summand, there is always at least one $\ell_j$ that is not $0$, and thus
\[
\beta_n = \frac{\Gamma\prn{1+1/d_1}}{\prn{a_1 \: n}^{1/d_1}} + \ltlo{\frac{1}{n^{1/d_1}}}.
\]
\end{proof}

\subsection{Proof of Claim~\ref{clm:ratio-upper-bound}}

\begin{repclaim}{clm:ratio-upper-bound}
For large enough $n$,
\[
\frac{\beta_n}{\beta_{n+1}} \sum_{\ell_1, \ell_2, \ldots \geq 0} \prod_{i = 1}^\infty \frac{\prn{-a_i \: \prn{\lambda(d_1) \beta_n}^{d_i} }^{\ell_i}}{\ell_i! \prn{1 + \sum_{j = 1}^\infty {d_j \ell_j}}} \leq 1.
\]
\end{repclaim}
\begin{proof}
By Lemma~\ref{lem:poly-hazard-beta}, we know that, for large enough $n$, there exist a constant $c \geq 0$ such that
\[
\beta_n = \frac{\Gamma\prn{1+1/d_1}}{\prn{a_1 \: n}^{1/d_1}} + \ltlo{\frac{1}{n^{1/d_1}}} \leq \frac{\Gamma\prn{1+1/d_1}}{\prn{a_1 \: n}^{1/d_1}} \prn{1 + \ltlo{1}}.
\]
Therefore, we have
\[
\frac{\beta_n}{\beta_{n+1}} = \prn{\frac{n+1}{n}}^{1/d_1} \prn{1 + \ltlo{1}}.
\]
Thus
\begin{gather}
\frac{\beta_n}{\beta_{n+1}} \sum_{\ell_1, \ell_2, \ldots \geq 0} \prod_{i = 1}^\infty \frac{\prn{-a_i \: \prn{\lambda(d_1) \beta_n}^{d_i} }^{\ell_i}}{\ell_i! \prn{1 + \sum_{j = 1}^\infty {d_j \ell_j}}} = \prn{1+\frac{1}{n}}^{1/d_1} \prn{1 + \ltlo{1}} \sum_{\ell_1, \ell_2, \ldots \geq 0} \prod_{i = 1}^\infty \frac{\prn{-a_i \: \prn{\lambda(d_1) \beta_n}^{d_i} }^{\ell_i}}{\ell_i! \prn{1 + \sum_{j = 1}^\infty {d_j \ell_j}}} \nonumber \\
\label{eq:higher-order} \leq \prn{1+\frac{1}{n}}^{1/d_1} \prn{1 + \ltlo{1}} \prn{1 - \sum_{i = 1}^\infty {\frac{a_i}{1 + d_i} \prn{\lambda(d_1) \beta_n}^{d_i}} + \ltlo{\frac{1}{n^{1/d_1}}}}.
\end{gather}
Notice that
\[
\sum_{i = 1}^\infty {\frac{a_i}{1 + d_i} \prn{\lambda(d_1) \beta_n}^{d_i}} = \frac{a_1}{1 + d_1} \prn{\lambda(d_1)}^{d_1} \beta^{d_1}_n + \sum_{i = 2}^\infty {\frac{a_i}{1 + d_i} \prn{\lambda(d_1) \beta_n}^{d_i}}.
\]
Also, $\beta^{d_i}_n = \bigO{\frac{1}{n^{d_i / d_1}}}$, and for $i \geq 2$, we have $d_i > d_1$, which implies that $\beta^{d_i}_n = \ltlo{\frac{1}{n^{1/d_1}}}$. Thus, \eqref{eq:higher-order} becomes
\begin{gather}
\frac{\beta_n}{\beta_{n+1}} \sum_{\ell_1, \ell_2, \ldots \geq 0} \prod_{i = 1}^\infty \frac{\prn{-a_i \: \prn{\lambda(d_1) \beta_n}^{d_i} }^{\ell_i}}{\ell_i! \prn{1 + \sum_{j = 1}^k {d_j \ell_j}}} \leq \prn{1+\frac{1}{n}}^{\frac{1}{d_1}} \prn{1 + \ltlo{1}} \prn{1 - \frac{a_1}{1 + d_1} \prn{\lambda(d_1) \beta_n}^{d_1} + \ltlo{\frac{1}{n^{1/d_1}}}} \nonumber \\
\leq \prn{\prn{1+\frac{1}{n}}^{1/d_1} + \ltlo{\frac{1}{n^{1/d_1}}}} \prn{1 - \frac{a_1}{1 + d_1} \prn{\lambda(d_1) \beta_n}^{d_1} + \ltlo{\frac{1}{n^{1/d_1}}}} \nonumber \\
= 1 + \frac{1}{d_1 \: n} - \frac{a_1}{1 + d_1} \prn{\lambda(d_1) \beta_n}^{d_1} + \ltlo{\frac{1}{n}} \nonumber \\
\label{eq:lambda-derivation} = 1 + \frac{1}{d_1 \: n} - \frac{a_1}{1 + d_1} \prn{\lambda(d_1)}^{d_1} \frac{\prn{\Gamma\prn{1+1/d_1}}^{d_1}}{a_1 n} + \ltlo{\frac{1}{n}}
\end{gather}
For \eqref{eq:lambda-derivation} to be bounded above by $1$, we need
\begin{gather*}
\frac{1}{d_1 \: n} \leq \frac{a_1}{1 + d_1} \prn{\lambda(d_1)}^{d_1} \frac{\prn{\Gamma\prn{1+1/d_1}}^{d_1}}{a_1 n} \iff \prn{\lambda(d_1)}^{d_1} \geq \frac{1+1/d_1}{\prn{\Gamma\prn{1+1/d_1}}^{d_1}} \iff \\
\lambda(d_1) \geq \frac{\prn{1+1/d_1}^{1/d_1}}{\Gamma\prn{1+1/d_1}},
\end{gather*}
which holds, since $\lambda(d_1) = \frac{\prn{1 + 1 / d_1}^{d_1}}{\Gamma\prn{1+1/d_1}}$.
\end{proof}

\subsection{Proof of Theorem~\ref{thm:hazard-constant-lower}}

\begin{reptheorem}{thm:hazard-constant-lower}
Consider the distribution $\cD$ for which $H(x) = x^d$ for $d \geq 0$. For any $\eps > 0$, there is no $\prn{\frac{\prn{1+1/d}^{1/d}}{\Gamma\prn{1+1/d}}-\eps}$-competitive cost prophet inequality for the single-item setting and I.I.D. random variables drawn from $\cD$.
\end{reptheorem}
\begin{proof}
Let $\lambda(d) = \frac{\prn{1+1/d}^{1/d}}{\Gamma\prn{1+1/d}}$. 

\begin{lemma}\label{lem:poly-hazard-beta-single-d}
For every $n \geq 1$,
\[
\beta_n = \frac{\Gamma\prn{1+1/d}}{n^{1/d}}.
\]
\end{lemma}
\begin{proof}
The proof follows immediately from the proof of Lemma~\ref{lem:poly-hazard-beta}. In particular, we have
\begin{equation}\label{eq:single-d-beta-calc-2}
\beta_n = \int^\infty_0 {e^{-n H(u)} \dif u} = \int^\infty_0 {e^{-n \: u^d} \dif u},
\end{equation}
and, by \eqref{eq:last-eq-of-betas} of Lemma~\ref{lem:poly-hazard-beta}, since $a_1 = 1$ and $a_2 = \dots = a_k = 0$, we get that
\[
\beta_n = \frac{\Gamma\prn{1+1/d}}{n^{1/d}}.
\]
\end{proof}

Using Lemma~\ref{lem:poly-hazard-beta-single-d}, we have that
\begin{equation}\label{eq:single-d-pre-sub}
R(n) = \frac{G(n)}{\beta_n} = \frac{n^{1/d}}{\Gamma(1+1/d)} \int^{G(n-1)}_0 {e^{-H(u)} \dif u} = \frac{n^{1/d}}{\Gamma(1+1/d)} \int^{G(n-1)}_0 {e^{-u^d} \dif u}.
\end{equation}
Let $x = u^d \iff u = x^{1/d}$. Also,
\[
\dif x = d \: u^{d - 1} \dif u \iff
\dif u = \frac{u^{1 - d}}{d} \dif x = \frac{x^{1/d-1}}{d} \dif x,
\]
and thus \eqref{eq:single-d-pre-sub} becomes
\begin{equation}\label{eq:r-expr-gamma}
R(n) = \frac{n^{1/d}}{d \: \Gamma(1+1/d)} \int^{\prn{G(n-1)}^d}_0 {e^{-x} x^{1/d - 1} \dif x} = \frac{n^{1/d}}{\Gamma(1+1/d)} \: \frac{1}{d} \: \gamma\prn{1/d, \prn{G(n-1)}^d}.
\end{equation}
where the second equality follows from the definition of the lower incomplete Gamma function.

\begin{lemma}\label{lem:ratio-incr}
$R(n)$ is increasing in $n$.
\end{lemma}
\begin{proof}
Recall that, by \eqref{eq:r-expr-gamma}, we have
\begin{align*}
R(n) &= \frac{n^{1/d}}{\Gamma(1/d)} \gamma\prn{1/d, \prn{G(n-1)}^d} \\
&= \frac{1}{d} \frac{n^{1/d}}{\Gamma\prn{1+1/d}} \gamma\prn{1/d, \prn{G(n-1)}^d} \\
&= \frac{1}{d \: \beta_n} \gamma\prn{1/d, \prn{G(n-1)}^d}.
\end{align*}
However, by Fact~\ref{fct:gamma-series}, we have
\[
\gamma\prn{1/d, \prn{G(n-1)}^d} = G(n-1) \sum_{k = 0}^\infty {\frac{\prn{- \prn{G(n-1)}^d}^k}{k! \: \prn{1/d + k}}}.
\]
Thus
\begin{align*}
R(n) &= \frac{G(n-1)}{d \: \beta_n} \: \sum_{k = 0}^\infty {\frac{\prn{- \prn{G(n-1)}^d}^k}{k! \: \prn{1/d + k}}} \\
&= \frac{G(n-1)}{\beta_{n-1}} \: \frac{\beta_{n-1}}{\beta_n} \: \sum_{k = 0}^\infty {\frac{\prn{- \prn{G(n-1)}^d}^k}{k! \: \prn{1 + d \: k}}} \\
&= R(n-1) \: \frac{\beta_{n-1}}{\beta_n} \: \sum_{k = 0}^\infty {\frac{\prn{- \prn{G(n-1)}^d}^k}{k! \: \prn{1 + d \: k}}}.
\end{align*}
It suffices to show that
\[
\frac{\beta_{n-1}}{\beta_n} \: \sum_{k = 0}^\infty {\frac{\prn{- \prn{G(n-1)}^d}^k}{k! \: \prn{1 + d \: k}}} \geq 1.
\]
Notice that
\[
\frac{\beta_{n-1}}{\beta_n} = \prn{\frac{n}{n-1}}^{1/d} = \prn{1 + \frac{1}{n-1}}^{1/d} = \sum_{\ell = 0}^{1/d} {\frac{1}{\prn{n-1}^\ell} \: \binom{1/d}{\ell}}.
\]
Thus, it suffices to show that
\[
\sum_{\ell = 0}^{1/d} {\frac{1}{\prn{n-1}^\ell} \: \binom{1/d}{\ell}} \cdot \sum_{k = 0}^\infty {\frac{\prn{- \prn{G(n-1)}^d}^k}{k! \: \prn{1 + d \: k}}} \geq 1.
\]
We use the fact that $G(n-1) \leq \lambda(d) \beta_{n-1} = \prn{\frac{1+1/d}{n-1}}^{1/d}$ and get
\begin{align*}
\sum_{\ell = 0}^{1/d} {\frac{1}{\prn{n-1}^\ell} \: \binom{1/d}{\ell}} \cdot \sum_{k = 0}^\infty {\frac{\prn{- \prn{G(n-1)}^d}^k}{k! \: \prn{1 + d \: k}}} &= \sum_{k = 0}^\infty {\sum_{\ell = 0}^{1/d} {\frac{1}{\prn{n-1}^\ell} \: \binom{1/d}{\ell}} \cdot \frac{\prn{- \prn{G(n-1)}^d}^k}{k! \: \prn{1 + d \: k}}} \\
&\geq \sum_{k = 0}^\infty {\sum_{\ell = 0}^{1/d} {\frac{1}{\prn{n-1}^\ell} \: \binom{1/d}{\ell}} \cdot \frac{\prn{- (1+1/d)}^k}{k! \: (n-1)^k \prn{1 + d \: k}}} \\
&= \sum_{k = 0}^\infty {\sum_{\ell = 0}^{1/d} {\binom{1/d}{\ell}} \frac{\prn{- (1+1/d)}^k}{k! \: \prn{1 + d \: k}}} \cdot \frac{1}{\prn{n-1}^{\ell+k}} \\
&= 1 + \frac{1}{d(n-1)} - \frac{1+1/d}{(d+1)(n-1)} + \bigO{\frac{1}{n^2}}
\end{align*}
Thus, for this quantity to be greater than $1$, it suffices to have
\[
\frac{1}{d(n-1)} \geq \frac{1+1/d}{(d+1)(n-1)} \iff \frac{d+1}{d} \geq 1+1/d,
\]
which is true.
\end{proof}

Assume, towards contradiction, that $\lim_{n \to \infty} R(n) = \lambda^*  < \lambda(d_1) = \frac{\prn{1+1/d}^{1/d}}{\Gamma\prn{1+1/d}}$.

We know that $G(n-1) = R(n-1) \beta_{n-1} = \frac{\Gamma\prn{1+1/d}}{(n-1)^{1/d}} R(n-1)$. Thus we get
\begin{equation}\label{eq:pre-gamma-fact}
R(n) = \frac{n^{1/d}}{\Gamma(1+1/d)} \: \frac{1}{d} \: \gamma\prn{1/d, \frac{\prn{\Gamma\prn{1+1/d}}^d}{(n-1)} \prn{R(n-1)}^d }.
\end{equation}
Recall that, by Fact~\ref{fct:gamma-series}, $\gamma(s, x) = x^s \sum_{k = 0}^\infty {\frac{\prn{-x}^k}{k! \: \prn{s+k}}}$, and thus \eqref{eq:pre-gamma-fact} becomes
\begin{align*}
R(n) &= \prn{\frac{n}{n-1}}^{1/d} R(n-1) \frac{1}{d} \sum_{k = 0}^\infty {\frac{\prn{- \frac{\prn{\Gamma\prn{1+1/d}}^d}{(n-1)} \prn{R(n-1)}^d}^k}{k! \: (1/d + k)}} \\
R(n) &= \prn{\frac{n}{n-1}}^{1/d} R(n-1) \sum_{k = 0}^\infty {\frac{\prn{- \frac{\prn{\Gamma\prn{1+1/d}}^d}{(n-1)} \prn{R(n-1)}^d}^k}{k! \: (1 + d \: k)}} \\
&= R(n-1) \prn{1 + \frac{1}{n-1}}^{1/d} \sum_{k = 0}^\infty {\frac{\prn{- \frac{\prn{\Gamma\prn{1+1/d}}^d}{(n-1)} \prn{R(n-1)}^d}^k}{k! \: (1 + d \: k)}} \\
\end{align*}
Notice that
\[
\prn{1 + \frac{1}{n-1}}^{1/d} = \sum_{\ell = 0}^{1/d} {\frac{1}{\prn{n-1}^\ell} \: \binom{1/d}{\ell}}.
\]
Thus,
\begin{align*}
& \prn{1 + \frac{1}{n-1}}^{1/d} \sum_{k = 0}^\infty {\frac{\prn{- \frac{\prn{\Gamma\prn{1+1/d}}^d}{(n-1)} \prn{R(n-1)}^d}^k}{k! \: (1 + d \: k)}} \\
= \sum_{\ell = 0}^{1/d} & {\frac{1}{\prn{n-1}^\ell} \: \binom{1/d}{\ell}} \cdot \sum_{k = 0}^\infty {\frac{\prn{- \frac{\prn{\Gamma\prn{1+1/d}}^d}{(n-1)} \prn{R(n-1)}^d}^k}{k! \: (1 + d \: k)}} \\
= \sum_{k = 0}^\infty & {\sum_{\ell = 0}^{1/d} {\frac{1}{\prn{n-1}^\ell} \: \binom{1/d}{\ell}} \cdot \frac{\prn{- \frac{\prn{\Gamma\prn{1+1/d}}^d}{(n-1)} \prn{R(n-1)}^d}^k}{k! \: (1 + d \: k)}} \\
= \sum_{k = 0}^\infty & {\sum_{\ell = 0}^{1/d} {\frac{1}{\prn{n-1}^{\ell+k}} \: \binom{1/d}{\ell}} \cdot \frac{\prn{- \prn{\Gamma\prn{1+1/d} \cdot  R(n-1)}^d}^k}{k! \: (1 + d \: k)}} \\
\approx 1 & + \frac{1}{d \: (n-1)} - \frac{\prn{\Gamma\prn{1+1/d} \cdot  R(n-1)}^d}{(d+1) \: (n-1)},
\end{align*}
where, for large enough $n$, we can ignore higher order terms and we also have $R(n-1) \approx \lambda^*$. Thus, for $R(n) \leq \lambda^*$, it must be that
\[
\frac{1}{d} - \frac{\prn{\Gamma\prn{1+1/d} \cdot \lambda^*}^d}{(d+1)} \leq 0 \iff \prn{\Gamma\prn{1+1/d} \cdot \lambda^*}^d \geq 1+1/d \iff \lambda^* \geq \frac{\prn{1+1/d}^{1/d}}{\Gamma\prn{1+1/d}},
\]
and we arrive at a contradiction.

Therefore, for any $\eps > 0$, there is no $\prn{\frac{\prn{1+1/d}^{1/d}}{\Gamma\prn{1+1/d}}-\eps}$-competitive cost prophet inequality for the single-item setting and I.I.D. random variables drawn from $\cD$.
\end{proof}

\subsection{Proof of Theorem~\ref{thm:single-threshold-lower}}

\begin{reptheorem}{thm:single-threshold-lower}
Consider the distribution $\cD$ for which $H(x) = x^d$ for $d \geq 0$. There is no $\ltlo{\prn{\log{n}}^{1/d}}$-competitive single-threshold cost prophet inequality for the single-item setting and I.I.D. random variables drawn from $\cD$.
\end{reptheorem}
\begin{proof}
Recall by \eqref{eq:single-alg-hazard} that
\[
R(n) = \frac{1}{\beta_n} \prn{\frac{1 - e^{-(n-1) H(T)}}{1 - e^{-H(T)}}
\prn{\int^T_0 {e^{-H(x)} \dif x} - T e^{-H(T)}} + e^{-(n-1) H(T)} \beta_1}.
\]
Assume, towards contradiction, that $R(n) = \ltlo{\prn{\log{n}}^{1/d}}$. For this to be the case, it must be that
\begin{equation}\label{eq:single-lower-bound-1}
e^{-(n-1) H(T)} \frac{\beta_1}{\beta_n} = \ltlo{\prn{\log{n}}^{1/d}},
\end{equation}
and also that
\begin{equation}\label{eq:single-lower-bound-2}
\frac{1}{\beta_n} \prn{\frac{1 - e^{-(n-1) H(T)}}{1 - e^{-H(T)}}
\prn{\int^T_0 {e^{-H(x)} \dif x} - T e^{-H(T)}}} = \ltlo{\prn{\log{n}}^{1/d}}.
\end{equation}
By \eqref{eq:single-lower-bound-1} and the definition of $\ltlo{\cdot}$, we have that for every $\eps > 0$, there must exist a $n_0 \geq 1$ such that for all $n \geq n_0$, we have
\begin{align}
e^{-(n-1) H(T)} \frac{\beta_1}{\beta_n} \leq \eps \prn{\log{n}}^{1/d} &\iff
e^{-(n-1) H(T)} \leq \eps \: \frac{\beta_n}{\beta_1} \: \prn{\log{n}}^{1/d} \iff \nonumber \\
-(n-1) H(T) \leq \log\prn{\eps \: \frac{\beta_n}{\beta_1} \: \prn{\log{n}}^{1/d}} &\iff
H(T) \geq \frac{\log\prn{\frac{\beta_1}{\eps \: \beta_n \: \prn{\log{n}}^{1/d}}}}{n-1} \iff \nonumber \\
\label{eq:single-lower-bound-3} T^d \geq \frac{\log\prn{\frac{\beta_1}{\eps \: \beta_n \: \prn{\log{n}}^{1/d}}}}{n-1} &\iff
T \geq \prn{\frac{\log\prn{\frac{\beta_1}{\eps \: \beta_n \: \prn{\log{n}}^{1/d}}}}{n-1}}^{1/d}.
\end{align}
However, by \eqref{eq:single-lower-bound-2}, we have that for every $\eps' > 0$, there must exist a $n_1 \geq 1$ such that for all $n \geq n_1$, we have
\begin{align}
\frac{1}{\beta_n} \prn{\frac{1 - e^{-(n-1) H(T)}}{1 - e^{-H(T)}}
\prn{\int^T_0 {e^{-H(x)} \dif x} - T e^{-H(T)}}} &\leq \eps' \: \prn{\log{n}}^{1/d} \nonumber \\
\prn{\frac{1 - e^{-(n-1) H(T)}}{1 - e^{-H(T)}}
\prn{\int^T_0 {e^{-x^d} \dif x} - T e^{-H(T)}}} &\leq \eps' \: \beta_n \: \prn{\log{n}}^{1/d} \nonumber \\
\label{eq:single-lower-bound-4} \prn{\frac{1 - e^{-(n-1) H(T)}}{1 - e^{-H(T)}}
\prn{\frac{1}{d} \gamma\prn{1/d, T^d} - T e^{-H(T)}}} &\leq \eps' \: \beta_n \: \prn{\log{n}}^{1/d}.
\end{align}
where the last equality follows by substituting $t = x^d$ in the integral, as seen several other times in the paper.

Notice that $T$ has to be decreasing in $n$, since, if not, one can easily see from \eqref{eq:single-alg-hazard-first} that the algorithm is too eager to select a value and its performance degrades rapidly as $n$ increases. Therefore, we know that $\lim_{n \to \infty} T = 0$. Furthermore, by Fact~\ref{fct:small-gamma-approx}, we know that for small $T$, i.e. large enough $n$, we have
\[
\gamma\prn{1/d, T^d} \approx d \: T,
\]
and thus \eqref{eq:single-lower-bound-4} becomes
\begin{align*}
\prn{\frac{1 - e^{-(n-1) H(T)}}{1 - e^{-H(T)}}
\prn{T - T e^{-H(T)}}} &\leq \eps' \: \beta_n \: \prn{\log{n}}^{1/d} \iff \\
\prn{\frac{1 - e^{-(n-1) H(T)}}{1 - e^{-H(T)}}
T\: \prn{1 - e^{-H(T)}}} &\leq \eps' \: \beta_n \: \prn{\log{n}}^{1/d} \iff \\
T \: \prn{1 - e^{-(n-1) H(T)}} &\leq \eps' \: \beta_n \: \prn{\log{n}}^{1/d}.
\end{align*}
However, by \eqref{eq:single-lower-bound-3} we know that we must have
\[
T \geq \prn{\frac{\log\prn{\frac{\beta_1}{\eps \: \beta_n \: \prn{\log{n}}^{1/d}}}}{n-1}}^{1/d},
\]
and if
\[
T \: \prn{1 - e^{-(n-1) H(T)}} \leq \eps' \: \beta_n \: \prn{\log{n}}^{1/d},
\]
then it also must be the case that
\[
T \: \prn{1 - \frac{\eps \: \beta_n \: \prn{\log{n}}^{1/d}}{\beta_1}} \leq \eps' \: \beta_n \: \prn{\log{n}}^{1/d}.
\]
Notice that by Lemma~\ref{lem:poly-hazard-beta}
\[
\beta_n = \frac{\Gamma\prn{1+1/d}}{n^{1/d}} \qquad \text{and} \qquad \beta_1 = \Gamma\prn{1+1/d},
\]
and thus
\[
T \: \prn{1 - \eps \: \prn{\frac{\log{n}}{n}}^{1/d}} \leq \eps' \: \beta_1 \: \prn{\frac{\log{n}}{n}}^{1/d}.
\]
For every $\eps$, for $n$ large enough, we have $1 - \eps \: \prn{\frac{\log{n}}{n}}^{1/d} > 0$, and thus
\begin{equation}\label{eq:single-lower-bound-5}
T \leq \eps' \: \frac{\beta_1 \: \prn{\frac{\log{n}}{n}}^{1/d}}{1 - \eps \: \prn{\frac{\log{n}}{n}}^{1/d}}.
\end{equation}
To arrive at a contradiction, we use \eqref{eq:single-lower-bound-3} and \eqref{eq:single-lower-bound-5} to show that it suffices to find, for every $\eps > 0$, a constant $\eps' > 0$ such that
\[
\eps' \: \frac{\beta_1 \: \prn{\frac{\log{n}}{n}}^{1/d}}{1 - \eps \: \prn{\frac{\log{n}}{n}}^{1/d}} < \prn{\frac{\log\prn{\frac{\beta_1}{\eps \: \beta_n \: \prn{\log{n}}^{1/d}}}}{n-1}}^{1/d} = \prn{\frac{\log\prn{\frac{1}{\eps} \: \prn{\frac{n}{\log{n}}}^{1/d}}}{n-1}}^{1/d}.
\]
Indeed, rearranging the terms above, we get
\begin{align*}
\eps'  &< \frac{1 - \eps \: \prn{\frac{\log{n}}{n}}^{1/d}}{\beta_1 \: \prn{\frac{\log{n}}{n}}^{1/d}} \cdot \prn{\frac{\log\prn{\frac{1}{\eps} \: \prn{\frac{n}{\log{n}}}^{1/d}}}{n-1}}^{1/d} \\
&= \frac{1}{\beta_1} \cdot \frac{1 - \eps \: \prn{\frac{\log{n}}{n}}^{1/d}}{\prn{\frac{\log{n}}{n}}^{1/d}} \cdot \prn{\frac{\frac{1}{d} \cdot \log\prn{\frac{1}{\eps^d} \: \frac{n}{\log{n}}}}{n-1}}^{1/d} \\
&= \frac{1}{\beta_1} \cdot \prn{1 - \eps \: \prn{\frac{\log{n}}{n}}^{1/d}} \cdot \prn{ \frac{\frac{1}{d} \cdot \frac{n}{\log{n}} \: \log\prn{\frac{1}{\eps^d} \: \frac{n}{\log{n}}}}{n-1} }^{1/d} \\
&= \frac{1}{\beta_1} \cdot \prn{1 - \eps \: \prn{\frac{\log{n}}{n}}^{1/d}} \cdot \prn{\frac{1}{d} \cdot \frac{n}{n-1} \cdot \frac{\log\prn{\frac{1}{\eps^d} \: \frac{n}{\log{n}}}}{\log{n}} }^{1/d} \\
&= \frac{1}{\beta_1 \: d^{1/d}} \cdot \prn{1 - \eps \: \prn{\frac{\log{n}}{n}}^{1/d}} \cdot \prn{\frac{n}{n-1} \cdot \frac{\log{n} - \log\prn{\eps^d \: \log{n}}}{\log{n}} }^{1/d}.
\end{align*}
Notice, however, that for any fixed $\eps > 0$, we have
\[
\lim_{n \to \infty} \prn{1 - \eps \: \prn{\frac{\log{n}}{n}}^{1/d}} \cdot \prn{\frac{n}{n-1} \cdot \frac{\log{n} - \log\prn{\eps^d \: \log{n}}}{\log{n}} }^{1/d} = 1,
\]
and thus, for every $\eps > 0$ there exists a large enough $n$ and a constant $0 < \eps' \frac{1}{\beta_1 d^{1/d}}$ such that \eqref{eq:single-lower-bound-3} and \eqref{eq:single-lower-bound-5} cannot simultaneously hold, and we arrive at a contradiction.
\end{proof}

\end{document}